\documentclass[%
 aapm,
 mph,%
 amsmath,amssymb,
 preprint,%
]{revtex4-2}

\setcitestyle{authoryear,round}
\usepackage{graphicx}
\usepackage{dcolumn}
\usepackage{amsthm,amsmath} 
\usepackage{algorithm}
\usepackage{algpseudocode}
\usepackage{bm}
\usepackage{comment}
\usepackage{hyperref}
\usepackage{etoolbox}

\makeatletter
\renewcommand\bibitem{\par\vspace{8pt}\@ifnextchar[\@lbibitem\@bibitem}
\makeatother

\def\E{{\sf{E}}}
\def\Pm{{\mathbb{P}}}

\def\Qm{{\mathbb{Q}}}

\def\R{{\mathbb{R}}}
\def\I{{\sf{I}}}
\def\Borel{{\cal{B}}}
\newtheorem{theorem}{Theorem}
\newtheorem{proposition}{Proposition}
\newtheorem{corollary}{Corollary}

\usepackage[mathlines]{lineno}

\usepackage{xcolor}

\begin{document}

\title{Bridge change-of-measure for propagator estimation of diffusions}
\author{\small Javier Aguilar}
\affiliation{\small Quantitative Life Sciences, The Abdus Salam International Centre for Theoretical Physics
(ICTP), Trieste, Italy \\
{\tt \small jaguilar@ictp.it} --
{\tt \small orcid.org/0000-0001-7842-6261} }
\author{\small Riccardo Gatto}
\affiliation{\small Institute of Mathematical Statistics and Actuarial Science, University of Bern, Alpeneggstrasse 22, 3012 Bern, Switzerland \\
{\tt \small riccardo.gatto@unibe.ch} --
{\tt \small orcid.org/0000-0001-8374-6964} }
\date{\small \today}
\begin{abstract}
{\small {\bf Abstract} \
The propagator of one-dimensional diffusions has relevant applications in statistical inference and across many disciplines. Yet it is obtained in closed form only for very few simple models. Here we analyze an importance sampling strategy based on a bridge change of measure to estimate propagators: paths are generated under an auxiliary diffusion whose bridge between the two fixed boundary points is easy to generate, and the target propagator is recovered through a Radon–Nikodym derivative or likelihood ratio correction derived from Girsanov’s theorem. This method overcomes the main limitations of the crude Monte Carlo estimator, which relies on binning in state space, introducing bias and becoming inefficient with rare transitions. We apply the method to a range of benchmark processes and show that the bridge-based estimator matches the accuracy of crude Monte Carlo with orders of magnitude fewer sample paths. Finally, we elaborate on discretization schemes for the Radon–Nikodym derivative and on approximations to the best-in-class change of measure, both of which increase the efficiency of the method.}
\\
{\small {\bf Keywords} \
Girsanov theorem --
Importance sampling --
Lamperti transform --
Likelihood ratio --
Markov property --
Ornstein-Uhlenbeck process --
Radon-Nikodym derivative --
Rare transition --
Wiener process} \\
{\small {\bf MSC2020} \ 60H10 Stochastic ordinary differential equations (aspects of stochastic analysis); 65C05 Monte Carlo methods}
\end{abstract}

\maketitle

\tableofcontents

\section{Introduction}

The propagator represents the major tool for the analysis of the distributional properties of Markov processes. Since it encapsulates all the probabilistic information about the evolution of the process, it 
is a central object in the physical, biological, economic and social sciences where such models are used~\cite{Van_Kampen_book,Risken_book,azaele_statistical_2016,Gardiner2009}. 
Moreover, the propagator is the important component for the construction of the likelihood of Markov models, making it essential for both Bayesian and maximum likelihood inference.

In particular, we will consider the propagator of diffusions described through stochastic differential equation (SDE) of the form
\begin{equation}\label{eq:general_SDE}
dX_t = A(X_t)\,dt + B(X_t)\,dW_t, \; t \ge 0,
\end{equation}
where $A,B$ are appropriately chosen real continuous (thus measurable) functions and where 
$W=\{W_t\}_{t \ge 0}$ is the standard Wiener process.
Eq.~\eqref{eq:general_SDE} as well as all other SDEs and stochastic integrals 
in this text are meant in the Itô sense.  

Two important instances of the SDE in Eq.
\eqref{eq:general_SDE} 
that admit the
strong solutions $X=\{X_t\}_{t \ge 0}$
are the following.
The geometric Brownian motion, with $A(x) = \mu x$,
$B(x) = D x$, for given $\mu \in \R,D > 0$,
is given by
$X_t = X_0 \exp\{(\mu - D^2/2) t + D W_t\}$.
The Ornstein-Uhlenbeck (OU) process, with $A(x) = - k x$,
$B(x) = D$, for given $k,D > 0$, is given by
\begin{align}                                           \label{e33}
X_t = e^{-k t} \left(X_0 +  D \int_0^t e^{k s} d W_s \right)
=e^{-k t} X_0 +  D \left(W_t-k\int_0^t e^{-k (t-s)} W_s 
d s \right);
\end{align}
the last formula demonstrating that it is clearly a strong 
solution. Further examples treated in this paper are listed in Table~\ref{tab:models}.

The propagator is defined as the conditional density 
of $X_t$ given $X_s$, $0 \le s < t$, namely through 
\begin{equation}\label{eq:def_propagator}
    \rho_{t|s}^X \left(x \mid y \right)=\lim_{h \to0} \frac{1}{h}
    {\Pm}\left[X_t \in[x,x+h) \big| X_s =y\right], \;
    x,y \in \R,
\end{equation}
provided $ {\Pm}\left[X_s \in [y,y+d y)\right]>0$.
The objective of this article is to estimate Eq. \eqref{eq:def_propagator}.
Because the process $X$ is Markovian,
the above conditional probability does not depend on events
generated by $\{X_u\}_{0\le u < s}$.
The propagator is characterized by the Fokker-Planck or forward Kolmogorov equation
\begin{equation}\label{eq:FPK-equation} 
\partial_t  \, \rho_{t|s}^X \left(x \mid y \right) = -\partial_x \left[A(x) \,\rho_{t|s}^X \left(x \mid y \right)\right] +\partial^2_x \left[\frac{B^2(x)}{2} \,\rho_{t|s}^X \left(x \mid y \right)\right],
\end{equation}
with the initial condition $\rho_{s\mid s}^X(x \mid y)=\delta(x-y)$, 
where $\delta$ 
stands for the Dirac delta function and
$\partial_z$ for the operator of partial differentiation 
with respect to
(w.r.t.) the real variable $z$.

The propagator is one of the central objects for inference on SDEs from 
 a sample 
 $(t_1,x_1),\ldots,(t_m,x_m)$, since the associated likelihood reads
\begin{equation}
    \mathcal{L}=\prod_{i=1}^{m-1}\rho^X_{t_{i+1}\mid t_{i}}(x_{i+1}\mid x_i);
\end{equation}
cf. ~\cite{sorensen_parametric_2004,craigmile_statistical_2023}.

The general issue with the definition of the propagator in Eq.
~\eqref{eq:def_propagator} and its PDE in 
Eq.~\eqref{eq:FPK-equation}
is that they lead to explicit solution in a few limited cases only. 
The inaccessibility of the closed-form expression of the propagator - and, consequently, of the likelihood - is widely regarded as one of the major 
problems of statistical inference with SDEs; cf. e.g. \cite{iacus_simulation_2008,sorensen_parametric_2004,craigmile_statistical_2023}. Classical approaches to circumvent this difficulty include approximating the propagator, by numerically integrating Eq.~\eqref{eq:FPK-equation}~\cite{Toral2014}, through expansions in orthogonal functions~\cite{ait-sahalia_maximum_2002} or by employing Monte Carlo (MC) methods. 

The aim of this work is to provide a rigorous analysis of a Monte Carlo method recently used in~\cite{PRX_Aguilar_2026} for statistical inference on model parameters. Specifically, we prove the main results on which the method rely and extend its range of applicability. This method belongs to the family of importance sampling or change-of-measure techniques and uses change-of-measure methods combined with  bridge processes and their propagators.

The structure of the article is the following.
In Sec.~\ref{sec:MC}, the Crude Monte Carlo propagator estimator is discussed.
In Sec.~\ref{sec:BCMMC}, the Bridge change-of-measure Monte Carlo 
estimator (BCMMCE) of the propagator is introduced.
Two discretization schemes for the evaluation of a path 
integral appearing in the BCMMCE are
presented in Sec.~\ref{sec:discretization_PI}.
On the practical side, Sec.~\ref{sec:algorithms} provides the two algorithms of Bridge change-of-measure Monte Carlo propagator estimation:
the estimation algorithm with Wiener bridges in Sec.~\ref{sec:algorithm} and
the estimation algorithm with arbitrary Markov bridges in Sec.~\ref{sec:algorithm_general}.
Then, Sec.~\ref{sec:Approximation} suggests a parametric approximation to the best-in-class bridge measure.
Sec.~\ref{VII} presents the results of the numerical study.
Applications of Wiener bridges to the Ornstein-Uhlenbeck, Cox-Ingersol-Ross and stochastic logistic target processes are presented in Sec.~\ref{sec:OU_w_wiener_Bridges}.
Applications of Ornstein-Uhlenbeck bridges to the Ornstein-Uhlenbeck and stochastic logistic target processes are given in Sec.~\ref{sec:varying_bridge}.
The approximation to a parametric best-in-class Ornstein-Uhlenbeck bridge is given in Sec.~\ref{sec:best_in_class_approx} and comparisons of path integral discretizations can be found in Sec.~\ref{sec:compare_PI_discretizations}.
Sec.~\ref{sVIII} concludes the article.

The following notation is used.\\
\noindent \hspace*{12mm}
$f'(s)= (d/d s) f(s)$, 
$f''(s)=  (d/d s)^2 f(s)$  \\
\noindent \hspace*{12mm} $\partial_t f(s,t) = (\partial/\partial t) f(s,t)$ \\
\noindent \hspace*{12mm}
$f^{(-1)}$: inverse or reciprocal of $f$  
\\
\noindent \hspace*{12mm} $\delta(x)$: 
Dirac delta function \\
\noindent \hspace*{12mm}
$\I\{S\}$: indicator of statement or set $S$ \\
\noindent \hspace*{12mm} $\E_\Qm$: expectation under probability 
measure $\Qm$ \\
\noindent \hspace*{12mm} $\Pm_Z$: probability measure
induced by stochastic process $Z$

\section{Crude Monte Carlo propagator estimation and its limitations}\label{sec:MC}

Before introducing the proposed 
bridge change-of-measure Monte Carlo (BCMMC) estimator, we first
introduce the underlying formal framework for SDEs and then summarize the 
steps for obtaining the sample mean or crude MC estimator of the propagator.

Consider $C_T \subset \R^{[0,T]}$
the space of continuous functions from $[0,T]$ to $\R$.
Associated to the distance
$d(\omega_1,\omega_2)= \sup_{t \in [0,T]} |\omega_1(t) - \omega_2(t)|$, $\forall \omega_1,\omega_2 \in C_T$,
the set $C_T$ forms a complete metric space.
Denote by $\Borel(C_T)$ the $\sigma$-algebra generated by the open sets of
$C_T$, that are obtained by the metric $d$.
Then $(C_T,\Borel(C_T))$ is a measurable space and
a stochastic process with continuous 
sample paths in this space is an element denoted 
$\omega = \omega_{[0,T]} \in C_T$. 
Let $0 \le s < t \le T$ and denote by
$\omega_{[s,t]} = \{\omega(u)\}_{s \le u \le t}$
the truncated process.
We then define the coordinate process
\begin{equation}                                            \label{e10}
W_t = W (t,\omega) = \omega(t) \, \text{ and } \,
W_{[s,t]} = \{W(u,\omega)\}_{s \le u \le t}=\omega_{[s,t]}.
\end{equation}
With this, $W=W_{[0,T]}=\{W_t\}_{t \in [0,T]}$
is a random element in $(C_T,\Borel(C_T))$. Kolmogorov extension
theorem tells that there exists a
probability measure ${\mathbb P}$ on
$(C_T,\Borel(C_T))$ under which $\omega_{[0, T]}$ or
the coordinate process
$W_{[0, T]}$
is the standard Wiener process with time domain $[0,T]$;
cf. p. 200-203 of \cite{shiryaev2016probability2}. Note that, although this construction can be extended to the time interval $[0,\infty)$, the analysis of absolute continuity of measures presented in the subsequent sections necessitates restricting attention to a finite time horizon $[0,T]$, where $T$ is arbitrary but fixed and may be chosen arbitrarily large.

The process \(X=X_{[0,T]}=\{X_t\}_{0\le t \le T}\), namely the diffusion implicitly defined through the SDE in Eq.~\eqref{eq:general_SDE} and 
restricted to $[0,T]$,
becomes a functional 
of the Wiener process $W_{[0, T]}$ taking the following form:
$X_t(\omega_{[0,t]})$.
Then, $X$ induces a new probability measure on the same
space $(C_T,\Borel(C_T))$ that we denote $\Pm_X$.
Define by $\Pm_{X;s,y}$ the conditional probability measure
\begin{align}                                           \label{e20}
\mathbb{P}_{X;s,y}[A]=\frac{\Pm_X[
 \{ \omega \in A \mid w_s \in [y,y+dy)\}]}
{\Pm_X[\{ \omega \in C_T \mid w_s \in [y,y+dy)\}]},
\end{align}
for any  $s \in [0, T]$, 
$A \in \Borel(C_T)$ and such that
$\Pm_X[\{ \omega \in C_T \mid w_s \in [y,y+dy)\}] \neq 0$.
The propagator of the process $X$ at the crossing points 
$(s,y)$, $(t,x) \in [0,T] \times \R$ is
$\rho_{t|s}^X \left(x \mid y \right)$,
given in Eq.
\eqref{eq:def_propagator}.
According to the context, we use the alternative notations
$
\rho_{t|s}^{\Pm_X} \left(x \mid y \right) =
\rho_{t|s}^X \left(x \mid y \right)$, 
this notation holding not only for the process $X$ but
also for the other processes considered.
The propagator can be re-expressed as
\begin{align}\label{eq:propagator_as_expected_value}
\rho_{t|s}^X \left(x \mid y \right) & = 
\int_{C_T} \delta(X_t(\omega_{[0,t]}) - x) \, d \Pm_{X;s,y}(\omega) 
= \E_{\Pm_{X;s,y}}\left[\delta(X_t - x)\right].
\end{align}
We will use the induced distribution of the process instead of
the process itself as upper index, whenever this will be required.
For later use, we note that Eq.~\eqref{eq:propagator_as_expected_value} gives
$\rho^X_{t|t} \left(x \mid x \right) = 1$ as well as
the alternative formula
$$\rho^X_{t|s} \left(x \mid y \right)  =
\E_{\Pm_X}\left[\delta(X_t - x) 
\mid X_s = y\right].$$

The sample mean or crude MC estimator (CMCE) of the propagator in Eq.~\eqref{eq:propagator_as_expected_value} is as follows: 
one draws a large number $n$ of trajectories from $\Pm_{X;s,y}$, using 
Euler-Maruyama's or Milstein's schemes, see e.g. ~\cite{Asmussen2007,Toral2014}, one replaces the Dirac delta function in Eq. \eqref{eq:propagator_as_expected_value} by the 
indicator of a small neighborhood of $x$, standardized by the size of the
neighborhood, and one takes the mean. This yields the CMCE
\begin{equation}\label{eq:MC_estimator_prop}
    \hat{\rho}_{t|s;n}^{X \; {\rm CMC}} (x|y) = \frac{1}{\Delta x} \frac{1}{n}\sum_{i=1}^n \I\left\{X^{(i)}_t\in\left[x,x+\Delta x\right)\right\},
\end{equation}
where $\I\{S\}$ denotes the indicator of statement $S$,
$X^{(i)}_t$ the $i^\text{th}$ sampling from $\Pm_{X;s,y}$,
for $i=1,\ldots,n$,
and $\Delta x>0$ is an appropriately small chosen discretization unit.
Due to the finiteness of $\Delta x$, this estimator is neither
unbiased nor consistent (as $n \to \infty$). According to
\cite{aguilar_biased_2024},
\begin{equation}
\E_{\Pm_{X;s,y}}\left[\hat{\rho}_{t|s;n}^{X \; {\rm CMC}} (x|y)  - \rho_{t\mid s} (x|y) \right]=\varepsilon(\Delta x),
\end{equation}
where $\varepsilon(\Delta x)$ is an unknown function that obeys $\lim_{\Delta x\to0}\varepsilon(\Delta x)=0$. Such bias can be particularly problematic in inference problems, where high-precision estimates of the propagator are required. Moreover, discretization units $\Delta x$ making $\varepsilon(\Delta x)$ sufficiently small are problem-dependent, so each application requires a careful 
analysis. A second issue is that estimating the propagator in rare transitions (i.e., where \(\mathrm{var}_{\Pm_{X;s,y}}\left(
\hat{\rho}_{t|s;n}^{X \; {\rm CMC}} (x|y)
 \right)\) can become excessively large
relative to $\{\rho_{t|s}^X \left(x \mid y \right)\}^2$, 
see \cite{Asmussen2007}) 
may require a prohibitively large number of generations in order to 
control the sampling error.
The method proposed in this article makes use of a bridge change-of-measure
in order to overcome these two limitations, as it completely avoids the use of discretizations and incorporates the necessary variance reduction needed to efficiently sample rare events.

\section{Bridge change-of-measure Monte Carlo propagator estimation}
\label{sec:BCMMC}

The main idea of importance sampling is to introduce an alternate
sampling process 
$Y=Y_{[0,T]}=\{Y_t\}_{0 \le t \le T}$ 
on the measurable space $(C_T,\Borel(C_T))$,
such that
$\Pm_X \ll \Pm_Y$, i.e. $\Pm_X$ is absolutely continuous w.r.t. 
to the induced probability measure of $Y$, $\Pm_Y$. In this case, the
Radom-Nikodym derivative $d\Pm_X/d \Pm_Y$ exists.
In addition, $Y$ or rather $\Pm_Y$ are chosen so to offer some
desirable sampling properties
that $\Pm_X$ does not; see e.g. ~\cite{Asmussen2007,gatto2015}.

Let us first remind the change-of-measure or likelihood ratio formula for diffusions,
commonly known as Girsanov's formula.
We recall that  $\Pm$ is the Wiener measure and
$W$ is its Wiener process, as given in Eq. \eqref{e10},
and that
$X=X_{[0,T]}=\{X_t\}_{0 \le t \le T}$ 
is the strong solution
of the SDE \eqref{eq:general_SDE} rewritten below,
\begin{equation*}
dX_t = A(X_t)\,dt + B(X_t)\,dW_t, \; 0 \le t \le T.
\end{equation*}
Assume that, for the real continuous function
$\tilde{A}$, the SDE
\begin{equation}                            \label{eq:general_SDE_V2}
dY_t = \tilde{A}(Y_t)\,dt + B(Y_t)\,dW_t, \; 0 \le t \le T ,
\end{equation}
admits the strong solution
$Y=Y_{[0,T]}=\{Y_t\}_{0 \le t \le T}$.
Then the following formula for the Radon-Nikodym derivative or 
likelihood ratio holds,
\begin{equation}                                \label{eq_girs}
   L^{X \mid Y}(Y_{[0,T]}) = \frac{d\Pm_X}{d\mathbb{P}_Y} 
   \Big|_{[0,T]}= \exp\left\{
    \int_0^T  \frac{A\left(Y_u\right)-\tilde{A}\left(Y_u\right)}{B^2\left(Y_u\right)}\,dY_u - \frac{1}{2} \int_0^T \frac{A^2\left(Y_u\right)-\tilde{A}^2\left(Y_u\right)}{B^2\left(Y_u\right)}\,du\right\},
\end{equation}
which 
is $\Pm_Y$-a.s. unique,
see e.g. Eq. (10.60) at p. 287 of \cite{klebaner_introduction_2012}.
This is Girsanov's formula, which is basically 
the exponential tilt of a Gaussian random vector;
see pp. 129-131 of \cite{Asmussen2007}.
Note that the fact that the SDEs in Eqs. ~\eqref{eq:general_SDE}
and \eqref{eq:general_SDE_V2} share the same volatility $B$ implies that
the measures $\Pm_X$ and $\Pm_Y$ are equivalent.
It is thus central to keep the same volatility, otherwise 
these two measures would be singular (namely probabilize disjoint subsets
of $C_T$).
Let $0 \le s < t \le T$ and $y \in \R$.
Denote by
$\Pm_{Y;s,y}$ the induced conditional probability of $Y$,
defined in analogy with Eq. \eqref{e20}.
The Change-of-measure theorem allows to
re-express the propagator in Eq.~\eqref{eq:propagator_as_expected_value} as
\begin{equation}\label{eq:propagator_change_of_measure}
\rho_{t|s}^X \left(x \mid y \right) =  \E_{\mathbb{P}_{Y;s,y}}\left[
\delta(X_t - x) \frac{d\Pm_{X;s,y}}{d\mathbb{P}_{Y;s,y}}
\Big|_{[s,t]} \right].
\end{equation}
It follows from the
Girsanov formula given in
Eq. \eqref{eq_girs} and from the Markov property that the Radon-Nikodym derivative inside the expectation of Eq. \eqref{eq:propagator_change_of_measure} is 
$\Pm_{Y;s,y}$-a.s. uniquely given by
\begin{equation}\label{eq:RN_der}
L^{X \mid Y}(Y_{[s,t]}) =
\frac{d\Pm_{X;s,y}}{d\mathbb{P}_{Y;s,y}} \Big|_{[s,t]}= 
   e^{(t-s)l^{X \mid Y}_{s,t}},
\end{equation}
with
\begin{equation}\label{eq:RN_der_2}
     l_{s,t}^{X \mid Y} = \frac{1}{t-s}
     \left(
     \int_s^t  \frac{A\left(Y_u\right)-\tilde{A}\left(Y_u\right)}{B^2\left(Y_u\right)}\,d Y_u - \frac{1}{2} \int_s^t \frac{A^2\left(Y_u\right)-\tilde{A}^2\left(Y_u\right)}{B^2\left(Y_u\right)}\,du \right).
\end{equation}

We now construct a change-of-measure induced by a bridge process, which proves to be particularly well-suited for the present setting.
By conditioning the induced measure $\Pm_Y$ to the points
$(s,a)$, $(t,b) \in [0,T] \times \R$, $s<t$, which amounts to force 
the process $Y$ to cross these two points,
one obtains the bridge version of $\Pm_Y$ denoted
$\mathbb{P}_{Y;s,a}^{t,b}$. 
This bridge probability measure $\mathbb{P}_{Y;s,a}^{t,b}$ is given by
\begin{align}                                \label{eq:bridge_meas}            
\mathbb{P}_{Y;s,a}^{t,b}[A]=
\frac{\Pm_Y[
 \{ \omega \in A \mid w_s \in [a,a+da), w_t \in [b,b+db)\}]}
{\Pm_Y[\{ \omega \in C_T \mid w_s \in [a,a+da), w_t \in [b,b+db)\}]},
\end{align}
for any $0 \le s < t \le T 
$, 
$A \in \Borel(C_T)$ and such that
$\Pm_X[\{ \omega \in C_T \mid w_s \in [a,a+da),
w_t \in [b,b+db)
\}] \neq 0$.
 Typically the bridge measure is characterized through its infinitesimal generator, which can be obtained as a Doob-h transform of the generator of the unconditional measure; see ~\cite{Asmussen2007,Chetrite2013,Chetrite2015}.

Clearly,
$\mathbb{P}_{Y;s,a}^{t,b}[A]> 0 \Rightarrow
\mathbb{P}_{Y;s,a}[A]>0$, $\forall A \in {\cal B}(C_T)$:
$\mathbb{P}_{Y;s,a}^{t,b} \ll \mathbb{P}_{Y;s,a}$ on $C_T$.
But for $\mathbb{P}_{Y;s,a} \ll \mathbb{P}_{Y;s,a}^{t,b}$ does
not hold, in general: the process $Y$ may not cross  
$(s,a)$ and $(t,b)$ but cross $(s,a)$ alone.
This leads us to consider a
randomization procedure
used in \cite{Aguilar_Gatto_2024} with discrete time Markov processes,
in which the fixed value $b$ at time $t$ of the bridge is replaced by a random
variable $B$. 
Assume that $B$ has density $g$ on $R$ and consider 
the mixed bridge measure $\tilde{\mathbb{P}}^{t,B}_{Y;s,a}$
characterized as follows. 
Under $\tilde{\mathbb{P}}^{t,B}_{Y;s,a}$,
the diffusion $Y$
conditional on $B=b$ has same distribution as $Y$
under ${\mathbb{P}}^{t,b}_{Y;s,a}$.
Without the condition $B=b$,
the diffusion $Y$ under $\tilde{\mathbb{P}}^{t,B}_{Y;s,a}$ passes:
through $a$ at time $s$ and, at time $t$, through 
the random variable $B$ with density $g$; this for
$0 \le s <t\le T$. Because $Y_t=B$ under 
$\tilde{\mathbb{P}}^{t,B}_{Y;s,a}$-a.s., then
$\tilde{\mathbb{P}}^{t,B}_{Y;s,a}$ prescribes the probability 
density $g$ to $Y_t$.
In the limiting case
\begin{equation}                                            \label{eq:forbidden}
g(x)=\delta(x-b),
\end{equation}
the two
measures
$\tilde{\mathbb{P}}^{t,B}_{Y;s,a}$ and
${\mathbb{P}}^{t,b}_{Y;s,a}$ coincide.
But this situation is the initial one, for which
$\mathbb{P}_{Y;s,a} \ll \mathbb{P}_{Y;s,a}^{t,b}$
does not hold.
A sufficient condition for the absolute continuity
$\mathbb{P}_{Y;s,a} \ll \tilde{\mathbb{P}}^{t,B}_{Y;s,a}$, namely
$\mathbb{P}_{Y;s,a}[A]>0 \Rightarrow \tilde{\mathbb{P}}^{t,B}_{Y;s,a}
[A]>0$, $\forall A \in {\cal B}(C_T)$, is that
the density $g$ is positive everywhere on $\R$.

These considerations lead to the practical Theorem \ref{th:RN_Q_Q^B_rand}.
\begin{theorem}[Bridge mixed likelihood ratio]\label{th:RN_Q_Q^B_rand}
Let $0 \le s < t \le T$ and $a,b\in \R$.
Let $B$ be a random variable with positive density $g$ on $\R$.
Let the diffusion $Y=Y_{[0,T]}$ be the strong solution to 
Eq.~\eqref{eq:general_SDE_V2}.
Consider the mixed bridge measure $\tilde{\mathbb{P}}^{t,B}_{Y;s,a}$
on $(C_T,\Borel(C_T))$ that,
conditional on $B=b$, 
coincides with the measure ${\mathbb{P}}^{t,b}_{Y;s,a}$. Then
$\mathbb{P}_{Y;s,a}$ and
$\tilde{\Pm}_{Y;s,a}^{t,B}$
are equivalent
and,
$\tilde{\mathbb{P}}_{Y;s,a}^{t,B}$-a.s.,
\begin{equation}\label{AP_eq:RN_Q_Q^B_V2}
\frac{d\mathbb{P}_{Y;s,a}}{d{\tilde{\mathbb{P}}}_{Y;s,a}^{t,B}}
\Big|_{[s,t]}=
\frac{\rho^{Y}_{t\mid s}(B\mid a)}{g(B)} =
\frac{\rho^{Y}_{t\mid s}(Y_t\mid a)}{g(Y_t)},
\end{equation}
\end{theorem}
\begin{proof} 
The first claim of equivalence of measures is 
explained just before this theorem.
Let $0 \le s < s' < s'' < t \le T$ and $a,b,x,y, \in R$, for the entire proof.
Since $\tilde{\mathbb{P}}^{t,B}_{Y;s,a}$ is a mixed Bridge measure, for any measurable functional $Z:C_T\to\mathbb{R}$ it holds
\begin{equation}\label{eq:mixture}
    \E_{\tilde{\mathbb{P}}^{t,B}_{Y;s,a}}[Z(Y)]=
    \int_{R} \E_{\tilde{\mathbb{P}}^{t,B}_{Y;s,a}}[Z(Y)\ \mid B = b ] g(b) db
    = \int_{R} \E_{\mathbb{P}^{t,b}_{Y;s,a}}[Z(Y)]g(b) db.
\end{equation}
Further, 
for any other arbitrary measurable functional $\tilde{Z}:C_T\to\mathbb{R}$ we have
\begin{align}\label{eq:identitH_th1}
    \E_{{\mathbb{P}}_{Y;s,a}}[\tilde{Z}(Y)]& =\int_{\R} \E_{{\mathbb{P}}_{Y;s,a}}[\tilde{Z}(Y)\mid
    Y_t=b ] \rho^{Y}_{t|s}(b|a) \,db  = \int_{\R} \E_{{\mathbb{P}}_{Y;s,a}^{t,b}}[\tilde{Z}(Y)]\,
    \rho_{t|s}^{Y}(b|a)\,db \nonumber \\ &
    =\int_{\R} \E_{{\mathbb{P}}_{Y;s,a}^{t,b}}[\tilde{Z}(Y)\rho^{Y}_{t|s}(b|a)] \,db 
     =\int_{\R} \E_{{\mathbb{P}}_{Y;s,a}^{t,b}}[\tilde{Z}(Y)\rho^{Y}_{t|s}(Y_t|a)] \,db \nonumber \\ &
     =\int_{\R} \E_{{\mathbb{P}}_{Y;s,a}^{t,b}}\left[\tilde{Z}(Y)\rho^{Y}_{t|s}(Y_t|a) \frac{g(b)}{g(b)}\right]  \,db
     =\int_{\R} \E_{{\mathbb{P}}_{Y;s,a}^{t,b}}\left[\tilde{Z}(Y)\frac{\rho^{Y}_{t|s}(Y_t|a)}{g(Y_t)}\right] g(b) \,db  \nonumber \\ & =
\E_{\tilde{\mathbb{P}}^{t,B}_{Y;s,a}}\left[\tilde{Z}(Y)\frac{\rho^{Y}_{t|s}(Y_t|a)}{g(Y_t)}\right],
     \end{align}
     where the last equality follows from Eq.~\eqref{eq:mixture}
     with $Z(Y) = \tilde{Z}(Y) \rho^{Y}_{t|s}(Y_t|a)/g(Y_t)$.
     Comparing 
     the Change-of-measure theorem
\begin{align*}
     \E_{{\mathbb{P}}_{Y;s,a}}[\tilde{Z}(Y)]& =
      \E_{\tilde{\mathbb{P}}^{t,B}_{Y;s,a}}\left[\tilde{Z}(Y) 
      \frac{d {\mathbb{P}}_{Y;s,a}}{d {\tilde{\mathbb{P}}}^{t,B}_{Y;s,a}}
      \Big|_{[s,t]}
      \right]
\end{align*}
with Eq.~\eqref{eq:identitH_th1} and considering that $\tilde{Z}$ is arbitrary, we obtain the likelihood ratio formula 
\begin{equation*}
\frac{d\mathbb{P}_{Y;s,a}}{d{\tilde{\mathbb{P}}}_{Y;s,a}^{t,B}} 
\Big|_{[s,t]} =
\frac{\rho^{Y}_{t\mid s}(Y_t\mid a)}{g(Y_t)},
\end{equation*}
which coincides with Eq.~\eqref{AP_eq:RN_Q_Q^B_V2}.
\end{proof}

A convenient choice for the density of $B$ is  
given 
by $g(x) = 
\phi(x;b,\sigma)$,
where
\begin{equation}                                    \label{eq:gauss}
\phi(x;\mu,\sigma)= \frac{1}{\sqrt{2\pi\sigma^2}}
e^{-\frac{(x-\mu)^2}{2\sigma^2}}, \quad
\forall x \in \R,
\end{equation}
is the Gaussian density with mean $\mu\in \R$ and standard deviation $\sigma>0$.
Because we can chose $\sigma$ as small as we wish,
we can be arbitrarily close to the situation with
the forbidden density \eqref{eq:forbidden}.

It is remarkable that the Radon-Nikodym derivative in 
Eq.~\eqref{AP_eq:RN_Q_Q^B_V2} depends only on the initial and final states of the process, that are $a$ and $B=Y_t$, and not on the whole path, as it happens in Eq.~\eqref{eq:RN_der}.

We can now show
Theorem \ref{tpropag}, which is our major result that 
relates the propagators under 
two different bridge measures.
\begin{theorem}[Bridge change-of-variable for propagator]    \label{tpropag}
Let
$0\le s < t \le T$ and $b,y \in \R$.
Let $\Pm_{X;s,y}$ and $\Pm_{Y;s,y}$ respectively 
be the measures induced by $X$ and $Y$, conditional on the crossing point
with bridge point $(s,y)$, as defined in \eqref{e20}.
Assume that
$\rho^{Y}_{t\mid s}(x|y)$
is a bounded and continuous function of $x\in \R$.
The propagators 
of $X$ and $Y$ are related through the change-of-variable formula
\begin{align}                           \label{eq:propagator_BCM}
    \rho^{X}_{t|s} \left(b \mid y \right)  =\E_{\Pm^{t,b}_{Y;s,y}}\left[\frac{d\Pm_{X;s,y}}{d\Pm_{Y;s,y}} \Big|_{[s,t]} \right]\rho^{Y}_{t\mid s}(b|y),
\end{align}
where $d\Pm_{X;s,y}/d\Pm_{Y;s,y} \mid_{[s,t]} $
is given in Eqs. \eqref{eq:RN_der} and 
\eqref{eq:RN_der_2}.
\end{theorem}
\begin{proof}
Let $0\le s < t < T$ and $b,y \in \R$.

We first note that 
$\Pm_{X;s,y}\ll \gg \Pm_{Y;s,y}$. This
follows directly from 
$\Pm_{X}\ll \gg \Pm_{Y}$ (given that $X$ and $Y$ share the
same volatility function, as mentioned just after Eq. \eqref{eq_girs}) 
and from the definition of these measure conditioned on $(s,y)$ in Eq. 
\eqref{e20}.

Then, let us re-express the propagator
as a limiting case.
Define
\begin{equation}
z_{\sigma}=\E_{\Pm_{X;s,y}}\left[\phi(X_t;b,\sigma)\right].
\end{equation}
We note that $\forall \omega \in \{ X_t \neq b \}$, the
sequence $\phi(X_t(\omega);b,n^{-1})$, $n=1,2,\ldots$,
decreases to 0. Because $\Pm_{X;s,y} [ X_t \neq b ] = 1$,
the Monotone convergence theorem
($\phi(X_t(\omega);b,(n+1)^{-1}) < 
\phi(X_t(\omega);b,n^{-1})$) 
and the propagator formula in
Eq.~\eqref{eq:propagator_as_expected_value} give us the limiting formula
\begin{equation}                                            \label{e93}
    \lim_{n \to \infty} z_{n^{-1}} =
    \E_{\Pm_{X;s,y}}\left[ \lim_{n \to \infty} \phi(X_t;b,n^{-1})\right] =
     \rho^X_{t|s}(b|y).
\end{equation}

We then apply the change-of-measure theorem, the chain rule and Theorem~\ref{th:RN_Q_Q^B_rand} to obtain
\begin{align*}
z_{\sigma}&=\E_{\tilde{\Pm}^{t,B}_{Y;s,y}}\left[\phi(X_t;b,\sigma)\frac{d\Pm_{X;s,y}}{d\tilde{\Pm}^{t,B}_{Y;s,y}} \Big|_{[s,t]} \right]
    =\E_{\tilde{\Pm}^{t,B}_{Y;s,y}}\left[\phi(X_t;b,\sigma)\frac{d\Pm_{X;s,y}}
    {d\Pm_{Y;s,y}} \Big|_{[s,t]} \frac{d\Pm_{Y;s,y}}{d\tilde{\Pm}^{t,B}_{Y;s,y}} \Big|_{[s,t]} \right]\nonumber \\
    &=\E_{\tilde{\Pm}^{t,B}_{Y;s,y}}\left[\frac{\phi(X_t;b,\sigma)}{g(X_t)}\frac{d\Pm_{X;s,y}}{d\Pm_{Y;s,y}}
    \Big|_{[s,t]}
    \rho^{Y}_{t\mid s}(X_t|y)\right],
    \end{align*}
given that $\Pm_{X;s,y}\ll\Pm_{Y;s,y}$.
As before, the function $g$ denotes the density of $B$, which is assumed
positive over $\R$.
Note that $\Pm_{Y;s,y}\ll 
\tilde{\Pm}^{t,B}_{Y;s,y}$
is given by Theorem~\ref{th:RN_Q_Q^B_rand}, which in turn implies
$\Pm_{X;s,y}\ll \tilde{\Pm}^{t,B}_{Y;s,y}$.
But the above result simplifies 
when using $g(x)=\phi(x;b,\sigma)$:
\begin{equation*}
    z_{\sigma}=\E_{\tilde{\Pm}^{t,B}_{Y;s,y}}\left[\frac{d\Pm_{X;s,y}}{d\Pm_{Y;s,y}} \Big|_{[s,t]} \rho^{Y}_{t\mid s}(X_t|y)\right].
\end{equation*}
Lastly, we take the limit $\sigma\to0$ in the last formula 
to obtain the following 
bridge change-of-measure formula for the propagator,
\begin{align}\label{eq:propagator_BCMa}
    \rho^{X}_{t|s} \left(b \mid y \right)  & =\lim_{\sigma\to0}z_{\sigma}  
    =\lim_{\sigma\to0}\E_{\tilde{\Pm}^{t,B}_{Y;s,y}}\left[\frac{d\Pm_{X;s,y}}{d\Pm_{Y;s,y}}
    \Big|_{[s,t]}
    \rho^{Y}_{t\mid s}(B|y)\right] \nonumber \\
    & =
    \lim_{\sigma\to0} \int_{-\infty}^\infty
    \E_{\tilde{\Pm}^{t,B}_{Y;s,y}}
    \left[\frac{d\Pm_{X;s,y}}{d\Pm_{Y;s,y}}
    \Big|_{[s,t]}
    \rho^{Y}_{t\mid s}(B|y) \Big| B=x \right]\phi(x;b,\sigma ) 
    dx \nonumber \\
& =
    \lim_{\sigma\to0} \int_{-\infty}^\infty
    \E_{{\Pm}^{t,x}_{Y;s,y}}
    \left[\frac{d\Pm_{X;s,y}}{d\Pm_{Y;s,y}}
    \Big|_{[s,t]}
    \rho^{Y}_{t\mid s}(x|y) \right]\phi(x;b,\sigma ) dx 
    \nonumber \\
& = \int_{-\infty}^\infty
    \E_{{\Pm}^{t,x}_{Y;s,y}}
    \left[\frac{d\Pm_{X;s,y}}{d\Pm_{Y;s,y}} \Big|_{[s,t]}
    \rho^{Y}_{t\mid s}(x|y) \right]\delta(x-b) dx \nonumber \\
    & =\E_{\Pm^{t,b}_{Y;s,y}}\left[\frac{d\Pm_{X;s,y}}{d\Pm_{Y;s,y}} \Big|_{[s,t]} \rho^{Y}_{t\mid s}(b|y)\right] =\E_{\Pm^{t,b}_{Y;s,y}}\left[\frac{d\Pm_{X;s,y}}{d\Pm_{Y;s,y}} \Big|_{[s,t]} \right]\rho^{Y}_{t\mid s}(b|y),
\end{align}
where the first equality is due to \eqref{e93},
the fourth to the definition of $\tilde{\Pm}^{t,B}_{Y;s,y}$ and
the fifth one to the characterization of convergence 
in distribution in terms of convergence of expectations of bounded continuous functions, which follows from the assumption that
$\rho^{Y}_{t\mid s}(x|y)$
is bounded and continuous w.r.t. $x$.
\end{proof}

Eq. \eqref{eq:propagator_BCM} leads to the 
definition of the BCMMCE of the propagator
\begin{align}                           \label{eq:propagator_BCMMCE}
    \hat{\rho}^{X}_{t|s} \left(b \mid y \right)  = \frac{d\Pm_{X;s,y}}{d\Pm_{Y;s,y}} 
    \Big|_{[s,t]}
    \hat{\rho}^{Y}_{t\mid s}(b|y),
\end{align}
for which $Y$ is sampled under the bridge measure $\Pm^{t,b}_{Y;s,y}$.

We note that the Dirac delta appearing in the first formula
Eq.~\eqref{eq:propagator_change_of_measure} disappears, so that the
discretization required by the CMCE in Eq.~\eqref{eq:MC_estimator_prop} 
disappears with the BCMMCE.
At first sight, Theorem \ref{tpropag} and 
the resuting BCMMCE in Eq. \eqref {eq:propagator_BCMMCE}
do not seem particularly useful, 
since computing the propagator of the process $X$ 
requires knowledge of the propagator of the process $Y$. However, 
this method gains relevance whenever $Y$ has a known or easy to evaluate propagator. 
Suitable choices are obtained by any process Gaussian process $Y$,
thus with Gaussian propagator, for which bridges can be sampled exactly; see ~\cite{Majumdar2015,aguilar_sampling_2022}. 
Precise suggestions are given
in Sec.
\ref{sec:algorithms}, which provides two generation algorithms:
the first one with Wiener bridges, which is a special
case of the second one with
more general bridges.

An obvious consequence of 
Theorem~\ref{tpropag}
is the following optimality result on the change-of-measure.
\begin{corollary}[Optimality of BCMMCE]\label{cor:optimallity}
Among all importance sampling estimators  
of the propagator 
$\rho_{t|s}^X \left(x \mid y \right)$, 
one with minimal zero variance
is the BCMMCE in Eq. \eqref{eq:propagator_BCMMCE},
whenever $\rho_{t|s}^Y \left(x \mid y \right)$ is available.
In particular, it is the BCMMCE with bridge sampling measure 
$\Pm_{X;s,y}^{t,x}$, which is defined in analogy with $\Pm_{Y;s,y}^{t,x}$
in Eq. \eqref{eq:bridge_meas}.
\end{corollary}
\begin{proof}
    Using $\Pm_{Y;s,y}=\Pm_{X;s,y}$ in Eqs.~\eqref{eq:propagator_BCM} 
    or~\eqref{eq:propagator_BCMMCE}
    yields an estimator with zero variance, since 
    $d\Pm_{X;s,y}/d\Pm_{Y;s,y}\mid_{[s,t]}=
    d\Pm_{X;s,y}/d\Pm_{X;s,y}\mid_{[s,t]}=1$. 
    Because any other minimal variance change-of-measure has variance 
    zero or more, 
    an optimal change-of-measure for estimating $\rho^X_{t\mid s}(x\mid y)$ is
    obtained from the bridge measure $\Pm_{X;s,y}^{t,x}$. 
\end{proof}

However, this optimal BCMMCE is irrelevant for the practice, because the argument is circular: evaluating Eq.~\eqref{eq:propagator_BCMMCE} with the optimal choice would require the very propagator we set out to estimate. Moreover,
generating bridges 
from $\Pm_{X;s,y}^{t,x}$  would also 
require knowledge of the propagator of the process $X$
(cf. ~\cite{Majumdar2015}).
Nevertheless, these results provide useful intuitive guidance for choosing the bridge sampling measure, as suggested in \cite{Asmussen2007}: selecting it as similar as possible to the one of the target process should lead 
to variance reduction

\section{Discretization of path integrals}\label{sec:discretization_PI}

Excepting few limited cases, computing the likelihood ratio
of the BCMMCE 
in \eqref{eq:propagator_BCMMCE}
requires the integration of the full path
$Y_{[s,t]}$, as seen
in Eq.~\eqref{eq:RN_der_2}. 
But the full path on the continuum is not accessible 
by simulation and only discretized versions,
usually referred to as discrete skeletons
(cf. p. 274 of ~\cite{Asmussen2007})
are available.
This section provides two 
discrete approximations to the integral in Eq.~\eqref{eq:RN_der_2}. 

The first solution is the direct application of the left-point
approximation to the integral.
With the homogeneous time partition $s=t_1<t_2<\dots <t_{m-1}
<t_m=t$, with 
mesh 
$\Delta t = t_{i+1}-t_i
=(t-s)/(m-1)$, for $i=1,\dots,m-1$, the skeleton of $Y_{[s,t]}$
is given by 
$y_i = Y_{t_i}(w_{[s,t]})$,
for $i=1,\dots,m$; $m \ge 2$.
This yields the following left-point approximation to Eq.~\eqref{eq:RN_der_2}, 
\begin{equation}\label{eq:RN_discretization}
    \tilde{l}_{s,t,\Delta t}^{X \mid Y}=\frac{1}{t-s}
     \sum_{i=1}^{m-1}  \frac{A\left(y_i\right)-\tilde{A}\left(y_i\right)}{B^2\left(y_i\right)}\,(y_{i+1}-y_i) - \frac{1}{2(m-1)} \sum_{i=1}^{m-1} \frac{A^2\left(y_i\right)-\tilde{A}^2\left(y_i\right)}{B^2\left(y_i\right)},
\end{equation}
provided $B(y_i) \neq 0$, for $i=1,\dots,m-1$, as
suggested by \cite{PRX_Aguilar_2026}. With this, 
the estimation of
$d\Pm_{X;s,y} / d\mathbb{P}_{Y;s,y} =  L^{X \mid Y}(Y_{[s,t]})$, given in
Eq.~\eqref{eq:RN_der}, becomes
\begin{equation}\label{eq:RN_discretization_V2}
    \tilde{L}_{\Delta t}^{X \mid Y}(\{y_i\}_{i=1,\dots,m}) = e^{(t-s)\tilde{l}_{s,t,\Delta t}}.
\end{equation}
Necessarily, approximating $L(Y_{[s,t]})$ 
in Eq.~\eqref{eq:propagator_BCM}
by $\tilde{L}_{\Delta t}(\{y_i\}_{i=1,\dots,m})$ 
introduces bias errors that tend to zero as $\Delta t\to 0$. Of course, 
this homogeneous discretization is not the only 
possibility
and the magnitude of the errors at finite $\Delta t$ 
depends on its choice. 

An specific discretization scheme can be obtained as follows. Define 
\begin{equation*}
    r(y)=\frac{A(y)-\tilde{A}(y)}{B^2(y)} \quad 
    \text{and} \quad q(y)=\frac{A^2\left(y\right)-\tilde{A}^2\left(y\right)}
    {B^2\left(y\right)}+r'(y)B^2(y).
\end{equation*}                                                               It\^o's lemma yields
\begin{equation*}
    d\left(\int^{Y_u}r(y)dy\right) = r(Y_u)dY_u+\frac{1}{2}r'(Y_u)B^2(Y_u)du.
\end{equation*}
This allows us to rewrite Eq.~\eqref{eq:RN_der_2} as
\begin{equation}\label{eq:RN_der_3}
     l_{s,t}^{X \mid Y} = \frac{1}{t-s}
     \left[
     \int_{Y_s}^{Y_t}  r(y)\,dy- \frac{1}{2} \int_s^tq\left(Y_u\right)du \right].
\end{equation}
Since Eq.~\eqref{eq:RN_der_3} contains only Riemann integrals, its discretization has the several advantages: The first integral is deterministic,since it only depends on the extremal points of the bridge. Moreover, in many cases this first integral can be evaluated exactly, hence it does not introduce 
neither a statistical nor a bias errors. 
The second integral in Eq.~\eqref{eq:RN_der_3} 
can be discretized using the trapezoidal rule, which has better convergence than the left-point approximation in Eq.~\eqref{eq:RN_discretization}. The resulting discretization scheme reads
\begin{equation}\label{eq:RN_discretization_V3}
    \hat{l}_{s,t,\Delta t}^{X \mid Y}=\frac{1}{t-s}
     \int_{y_1}^{y_m}  r(y)\,dy- \frac{1}{4(m-1)} \sum_{i=1}^{m-1} \left[q(y_{i+1})+q(y_{i})\right],
\end{equation}
Hence an alternative estimator to Eqs.
\eqref{eq:RN_discretization}
and \eqref{eq:RN_discretization_V2},
for the Radon-Nikodym derivative
in Eq.~\eqref{eq:propagator_BCM},
is given by Eq. \eqref{eq:RN_discretization_V3} and
\begin{equation}\label{eq:RN_discretization_V4}
    \hat{L}_{\Delta t}^{X \mid Y}(\{y_i\}_{i=1,\dots,m}) = e^{(t-s)\hat{l}^{X \mid Y}_{s,t,\Delta t}}.
\end{equation}

\section{Algorithms of
Bridge change-of-measure Monte Carlo propagator estimation} \label{sec:algorithms}

Theorem~\ref{tpropag} yields the algorithms presented in this section. As noted after Theorem~\ref{tpropag}, the BCMMC technique is particularly well suited to sampling processes $Y$ whose bridges are easy to simulate and whose propagators are known. In practice, however, requiring $Y$ to share the volatility function $B$ of the target process $X$ is often too restrictive. We therefore introduce two auxiliary processes $K$ and $H$ on $(C_T, \mathcal{B}(C_T))$.
The process $K$ has an arbitrary volatility function $\beta$. 
The bridges of $K$ are easy to sample.
The process
$H$ shares the volatility function of $K$, by construction.
Precisely,
$H_t = h(X_t)$, where the transformation $h$ is chosen so that $H$ has volatility function $\beta$.
We then proceed in three steps: (i) sample bridges of $K$;
(ii) use these bridges to estimate the propagator of $H$; (iii) transform this estimator into the propagator estimator of $X$. 

For didactic reasons, Sec.~\ref{sec:algorithm} first treats the case $K = W$, the Wiener process, for which $\beta \equiv 1$ and $h$ is the classical Lamperti transform. Sec.~\ref{sec:algorithm_general} then extends the procedure to a general sampling process $K$ with volatility function $\beta$, for which $h$ is a generalized Lamperti transform.
 
\subsection{Estimation algorithm with Wiener bridges}\label{sec:algorithm}

A basic choice of sampling process $K$ for the bridge construction
is the Wiener one $K=W$, whose propagator takes the 
simple form
\begin{equation}                                            \label{eq:wienprop}
\rho^{W}_{t\mid s}(x|y)=\phi(x;y,t-s),
\end{equation}
$\phi$ being the Gaussian density in Eq. ~\eqref{eq:gauss}. 
The corresponding Wiener bridge 
with extremal points $(s,a)$ and $(t,b)$
admits 
the SDE
\begin{equation}\label{eq:W-bridge}
d\overline{W}_u=-\frac{\overline{W}_u-b}{t-u}\,du+dW_u,
\quad 0 \le s < u < t \le T,\; a,b \in \R,
\end{equation}
with $\overline{W}_s=a$;
cf. \cite{klebaner_introduction_2012} p. 133. Importantly, since the Wiener bridge $\overline{W}$ remains Gaussian, it can be sampled exactly: see~\cite{Majumdar2015,aguilar_sampling_2022,PRX_Aguilar_2026}.
In principle, one could only make direct use of such Wiener bridges for the 
BCMMC technique in conjunction with any process $X$ that is additive (in 
the sense of constant volatility) and 
with unit volatility $B(x)=1$, $\forall x \in \R$, 
which is the one of $W$. Should the volatility function $B$ differ,
then the likelihood ratio could not be obtained through 
Girsanov's formula in Eqs.~\eqref{eq:RN_der} and \eqref{eq:RN_der_2}.
In the context of our one-dimensional processes, this restriction can be readily circumvented through the Lamperti transform; cf. e.g. p. 77 of ~\cite{pavliotis_stochastic_2014}. It transforms the target process $X$ of Eq.~\eqref{eq:general_SDE} to some additive process
$H_t= h(X_t)$, for some Borel function $h$.
This transform is any primitive of 
$1/B$, 
\begin{equation}\label{eq:Lamperti_transformation}
    h(x)=\int^x \frac{ds}{B(s)},
\end{equation}
whenever it exists.
We require $h$ injective, twice
differentiable and with continuous second derivative $h''(x)=-B'(x)/B^2(x)$.
These conditions hold iff 
$B$ has positive range and continuous first derivative
$B'$. These conditions are fulfilled by
all common diffusions.
It follows from Itô’s formula that the dynamics of \( H \) follows the SDE
\begin{equation*}
    dH_t=A_h(H_t)\,dt+dW_t,
    \; 0 \le t \le T,
\end{equation*}
where the drift function \( A_h \) takes the form
\begin{equation}\label{eq:transformed_drift}
    A_h(y)=h'\!\left(h^{(-1)}(y)\right)A\!\left(h^{(-1)}(y)\right)
    +\frac{1}{2} B^2\!\left(h^{(-1)}(y)\right) h''\!\left(h^{(-1)}(y)\right).
\end{equation}
Then, the propagators of 
the process $X$ and of its Lamperti transform $H$, 
which is denoted $\rho^H_{t\mid s}(x|y)$,
satisfy the propagator change-of-variable formula
\begin{equation}\label{eq:relation_between_propagators}
    \rho^{X}_{t\mid s}(x|y)=\rho^{H}_{t\mid s}\left(h(x)|h(y)\right) \; |h'\left(x\right)|.
\end{equation}
The BCMMC method starts by estimating $\rho^{H}_{t\mid s}\left(h(x)|h(y)\right)$ with the change-of-measure formulae of 
Sec.~\ref{sec:BCMMC}, 
with
the measure $\Pm_H$ replacing $\Pm_X$ and
the Wiener measure $\Pm$ replacing $\Pm_Y$.
First, Girsanov's formulae in Eqs.~\eqref{eq:RN_der}
and~\eqref{eq:RN_der_2}
with $A=A_h$, $\tilde{A}=0$ and $B=1$
give
\begin{equation}                                \label{eq_girs_H}
   L^{H \mid W}(W_{[s,t]}) = \frac{d\Pm_{H;s,y}}{d\Pm_{s,y}} 
   \Big|_{[s,t]}
   = \exp\left[
    \int_s^t  A_h\left(W_u\right)dW_u - \frac{1}{2} \int_s^t A_h^2\left(W_u\right)du\right],
\end{equation}
where $d\Pm_{H;s,y}$ and $d\Pm_{s,y}$ are
defined in analogy with Eq. \eqref{e20}.
This is used with
the change-of variable formula in 
Eq.~\eqref{eq:propagator_BCM} and the Wiener propagator formula~\eqref{eq:wienprop}, leading to
\begin{equation}\label{eq:propagator_BCM_V2}
    \rho^{H}_{t|s} \left(x \mid y \right) =\E_{\Pm^{t,x}_{s,y}}
    \left[\frac{d\Pm_{H;s,y}}{d\Pm_{s,y}} \Big|_{[s,t]}
    \right]
    \rho^{W}_{t\mid s}(x|y) = \E_{\Pm^{t,x}_{s,y}}
    \left[\frac{d\Pm_{H;s,y}}{d\Pm_{s,y}} \Big|_{[s,t]}
    \right]
    \phi(x;y,t-s).
\end{equation}
Numerical methods have only access to stochastic skeletons $\{\overline{w}_i\}_{i=1\dots,m}$ of $\overline{W}_{[s,t]}$, 
with $\overline{w}_i = \overline{W}_{t_i}(\omega_{[s,t_i]})$ and $t_i = s+ (i-1) \Delta t$, 
for $i=1, \ldots, m$.
We apply the left-point discretization formulae of Sec.~\ref{sec:discretization_PI} to evaluate the path integrals over stochastic skeletons. Thus,
Eqs.~\eqref{eq:RN_discretization} 
and~\eqref{eq:RN_discretization_V2} lead to
\begin{equation}\label{eq:RN_discretization_H}
    \tilde{l}^{H \mid {W}}_{s,t,\Delta t}=\frac{1}{t-s}
     \left( 
     \sum_{i=1}^{m-1} A_h\left(\overline{w}_i\right)\,(\overline{w}_{i+1}-\overline{w}_i) 
     - \frac{1}{2(m-1)} \sum_{i=1}^{m-1} A_h^2\left(\overline{w}_i\right) \right),
\end{equation}
and
\begin{equation}\label{eq:RN_discretization_H_V2}
    \tilde{L}^{H\mid W}_{\Delta t}(\{\overline{w}_i\}_{i=1,\dots,m}) = e^{(t-s)\tilde{l}^{H\mid W}_{s,t,\Delta t}}.
\end{equation}
Finaly, one applies Eq. \eqref{eq:relation_between_propagators}.

These developments lead to
Algorithm~\ref{al:ACM} for BCMMCE of
of propagators of one-dimensional processes using the Wiener bridge.
\begin{algorithm}[H]
  \caption{ BCMMCE of $\rho^X_{t\mid s}(x|y)$ for $X$ solving Eq.~\eqref{eq:general_SDE} with Wiener bridge $\overline{W}$
  solving Eq.~\eqref{eq:W-bridge}}
  \label{al:ACM}
   \begin{algorithmic}[1]
    \State Perform the Lamperti transform on data, $\tilde{x}=h(x)$, and $\tilde{y}=h(y)$, where $h$ is given by Eq.~\eqref{eq:Lamperti_transformation}.
    \State Generate $\overline{n}$ Wiener bridges through Eq.~\eqref{eq:W-bridge} 
    with discretization $\overline{\Delta t} = (t-s)/m$, $m\in\mathbb{N}^{\ge3}$, connecting the points ($s,\tilde{y}$) and ($t,\tilde{x}$). These 
    generated
    sample paths are the skeletons, $\{\overline{w}_i^{(j)}\}_{i=1,\dots,m}$, $j=1,\dots,\overline{n}$, 
    that will be used to evaluate Eq.~\eqref{eq:RN_discretization_H}. 
    \State Build the transformed drift function $A_h$ through Eq.~\eqref{eq:transformed_drift}, which itself requires 
    the functions
    $h'$, $h''$, and $h^{(-1)}$.
    \State Using the bridge skeletons and the transformed drift function 
    from Steps 2 and 3, 
    evaluate the estimator
    $\tilde{L}^{H\mid W}_{\overline{\Delta t}}(\{\overline{w}_i\}_{i=1,\dots,m})$
    of the likelihood ratio at each one of the $\overline{n}$ generated bridges, according to Eqs.~\eqref{eq:RN_discretization_H} and~\eqref{eq:RN_discretization_H_V2} with $\Delta t=\overline{\Delta t}$, 
    yielding
    $
    \left\{\tilde{L}^{^{H\mid W}(j)}_{\overline{\Delta t}}\right\}_{j=1,\ldots,\overline{n}}$.
    \State Estimate the propagator of the transformed process using Eq.~\eqref{eq:propagator_BCM_V2}:
    \begin{equation}
        \hat{\rho}^{H}_{t\mid s}(\tilde{x}|\tilde{y})=\frac{\phi(\tilde{x};\tilde{y},t-s)}{\overline{n}}\sum_{j=1}^{\overline{n}}\tilde{L}_{\overline{\Delta t}}^{^{H\mid W}(j)}.
    \end{equation}
    \State Finally, obtain the BCMMCE of
    the propagator of the target process $X$ by 
    the change-of-variable formula Eq.~\eqref{eq:relation_between_propagators}: $$\hat{\rho}_{t|s;\overline{n},\overline{\Delta t}}^{X \; \text{BCM}} (x|y) =\hat{\rho}^{H}_{t\mid s}(\tilde{x}|\tilde{y}) \,|h'(x)|.$$ 
   \end{algorithmic}
\end{algorithm}

We observe that when the target process $X$ 
 shares the same noise function as the bridge process $\overline{W}$, 
 i.e. \(B(x) = 1\), Steps 1 and 6 are superfluous and can be omitted: in this case 
 $h(x) = x$ yielding $H=X$, $\tilde{x}=x$, $\tilde{y}=y$, and $\hat{\rho}^{H}_{t\mid s}(\tilde{x}|\tilde{y})=\hat{\rho}^{X}_{t\mid s}(\tilde{x}|\tilde{y})$. 

The estimator obtained via Algorithm~\ref{al:ACM} is affected by two main sources of error. The first is statistical, governed by the number of 
generated bridge sample paths \( \overline{n} \). The second is a systematic bias arising from the fact that the generated trajectories are discrete-time skeletons 
of the underlying continuous process. 
Therefore they depend on the temporal discretization unit \( \Delta t \), as discussed in Sec.~\ref{sec:discretization_PI}. Nevertheless, both sources of error are under control: by increasing \( \overline{n} \) and refining the time step \( \overline{\Delta t} \), the overall error can be made arbitrarily small, at the cost of increased computational effort. 

\subsection{Estimation algorithm with arbitrary  bridges}\label{sec:algorithm_general}

The estimation of propagators via a Wiener bridge change-of-measure can be readily generalized to arbitrary bridges. The procedure adheres to the framework presented in Sec.~\ref{sec:algorithm}. The first step consists in transforming the process $X$ to some process denoted $H$ 
so that the noise function $\beta(x)$ 
coincides with the one of the new process $K$
on $(C_T, {\cal B}(C_T))$, which is used to create the 
bridge and which has a known
propagator. The process underlying the bridge is the strong solution of the SDE
\begin{equation}                                \label{eq:SDE_for_bridges}
    dK_u = \alpha(K_u)du+\beta(K_u)\,dW_u, \; \quad 0 \le t \le T,
\end{equation}
where $\alpha$ and $\beta$ are continuous functions.
With $\alpha(k)=0$ and 
$\beta (k) = 1$, we have 
$K=W$ as in Sec.~\ref{sec:algorithm}.
Building bridges 
through the points $(s,a)$ and $(t,b)$
from the process $K$ 
of Eq.~\eqref{eq:SDE_for_bridges} is theoretically feasible through $h$-Doob transforms, see e.g.~\cite{jamison_markov_process_of_schrodinger,Asmussen2007,Majumdar2015,aguilar_sampling_2022},
although numerical efficiency is only guaranteed for processes in which the propagator is known or easy to compute (e.g. Gaussian processes or the geometric Brownian motion). 
This leads to the bridge SDE
\begin{equation}                                \label{eq:SDE_for_bridges2}
    d\overline{K}_u = \overline{\alpha}(\overline{K}_u,u)du+
    \overline{\beta}(\overline{K}_u)\,dW_u, \; \quad 0 \le s < u < t \le T,
\end{equation}
where $\overline\alpha(k,u) = \alpha(k) + \beta^2(k)\,\partial_k \log \rho^{K}_{t|u}(b \mid k)$, and $\overline\beta(k)=\beta(k)$.
Eq. \eqref{eq:W-bridge} is the instance of
Eq. \eqref{eq:SDE_for_bridges2} with
$\overline{\alpha}(k,u)= 0+ 1^2 \;
\partial_k \log \phi(k;b,t-u) = 
- (k- b)/(t-u)$ and 
$\overline{\beta} (k) = 1$.

As stated before, one needs to transform the process $X$ of Eq.~\eqref{eq:general_SDE} 
to some other process $H$,
in order to use the general
bridge $\overline{K}$ of Eq.~\eqref{eq:SDE_for_bridges2}. 
We use the generalization of the Lamperti transform 
$H_t = h(X_t)$ 
obtained from
\begin{equation}\label{eq:generalized_lamperti}
    h(x)=\int^x\frac{\beta(s)}{B(s)}ds,
\end{equation}
whenever the integral exists,
with $B$ representing the volatility of $X$, in Eq.~\eqref{eq:general_SDE}.
We require $h$ injective, twice
differentiable and with continuous second derivative 
$h''(x)=\{\beta'(x) B(x) - \beta(s) B'(x)\}/B^2(x)$.
These conditions hold if 
$B$ and $\beta$ have positive range and continuous first derivatives
$B'$ and $\beta'$. 
The Lamperti transform in 
Eq.~\eqref{eq:Lamperti_transformation} is retrieved with $\beta = 1$. 
By It\^o's lemma, the SDE of $H_t=h(X_t)$ reads
\begin{equation}                                \label{newprocess}
    dH_t=\alpha_h(H_t)\,dt+\beta(H_t)\,dW_t,
    \; 0 \le t \le T.
\end{equation}
where
\begin{equation}\label{eq:transformed_drift_F}
    \alpha_h(y)=h'\!\left(h^{(-1)}(y)\right)A\!\left(h^{(-1)}(y)\right)
    +\frac{1}{2} h''\!\left(h^{(-1)}(y)\right) \,
    B^2\!\left(h^{(-1)}(y)\right),
\end{equation}
with  $h' = \beta/B$.
Thus the SDEs in Eqs.~\eqref{eq:SDE_for_bridges} and~\eqref{newprocess} have the common volatility function $\beta$.

The last part of the procedure is analogous to that presented in  Sec.~\ref{sec:algorithm}, differing only by minor modifications in the principal formulae. 
We consider the 
Radon-Nikodym derivative in Eqs.~\eqref{eq:RN_der} and
\eqref{eq:RN_der_2}
with the stochastic process $X$ replaced by $H$ and the process $Y$ 
replaced by $K$. 
Thus $A=\alpha_h$, $\tilde{A}=\alpha$ and $B=\beta$, leading to
\begin{equation}    \label{eq_girs_l}
l^{H \mid K}_{s,t} =
    \int_s^t  \frac{\alpha_h\left(K_u\right)-\alpha\left(K_u\right)}{\beta^2\left(K_u\right)}\,dK_u - \frac{1}{2} \int_s^t \frac{\alpha_h^2\left(K_u\right)-\alpha^2\left(K_u\right)}{\beta^2\left(K_u\right)}\,du
\end{equation}
and
\begin{equation}    \label{eq_girs_F}
L^{H \mid K}(K_{[s,t]}) =
\frac{d \Pm_{H;s,y}}{d \Pm_{K;s,y}} \Big|_{[s,t]} =
     \exp \left\{
    l^{H \mid K}_{s,t}
    \right\}.
\end{equation}
With the change-of-variable formula in 
Eq.~\eqref{eq:propagator_BCM} we obtain the propagator of the 
transformed process,
\begin{equation}\label{eq:propagator_BCM_general}
    \rho^{H}_{t|s} \left(x \mid y \right) = \E_{\Pm^{t,x}_{K;s,y}}
    \left[\frac{d\Pm_{H;s,y}}{d\Pm_{K;s,y}} \Big|_{[s,t]}
    \right]
    \rho^{K}_{t\mid s}(x|y),
\end{equation}
where
$\Pm_{H;s,y}$ and $\Pm_{K;s,y}$ are defined like in
Eq. \eqref{e20}.
We apply the left-point discretization formulae of Sec.~\ref{sec:discretization_PI}, that are
Eqs.~\eqref{eq:RN_discretization} 
and~\eqref{eq:RN_discretization_V2}. 
Denote $\overline{k}_i = \overline{K}_{t_{i}}(\omega_{[0,t_i]})$, for $i=1,\dots,m-1$, 
then the respective discretizations 
of Eqs.~\eqref{eq_girs_l} and \eqref{eq_girs_F}
are given by
\begin{equation}\label{eq:RN_discretization_F}
    \tilde{l}^{H \mid K}_{s,t,\Delta t}=\frac{1}{t-s}
     \left(
     \sum_{i=1}^{m-1}  \frac{\alpha_h\left(\overline{k}_i\right)-\alpha\left(\overline{k}_i\right)}{\beta^2\left(\overline{k}_i\right)}\,(\overline{k}_{i+1}-\overline{k}_i) - \frac{\Delta t}{2} \sum_{i=1}^{m-1} \frac{\alpha_h^2\left(\overline{k}_i\right)-{\alpha}^2\left(\overline{k}_i\right)}{\beta^2\left(\overline{k}_i\right)} \right)
\end{equation}
and
\begin{equation}\label{eq:RN_discretization_F_V2}
    \tilde{L}^{H\mid K}_{\Delta t}(\{\overline{k}_i\}_{i=1,\dots,m}) = e^{(t-s)\tilde{l}^{H\mid K}_{s,t,\Delta t}}.
\end{equation}
Finally, we use the result that
the propagators of the processes $X$ and $H$
are
related as before through the change-of-variable formula
\eqref{eq:relation_between_propagators}.

These formulae lead to Algorithm~\ref{al:ACM_V2} for 
the BCMMC estimator of
of the propagator of a one-dimensional process using a general bridge.
\begin{algorithm}[H]
  \caption{ BCMMCE of $\rho_{t\mid s}^X(x|y)$ for $X$ solving Eq.~\eqref{eq:general_SDE} with general bridges $\overline{K}$ solving Eq.~\eqref{eq:SDE_for_bridges2}}
  \label{al:ACM_V2}
   \begin{algorithmic}[1]
    \State Perform the generalized Lamperti transform 
    $\tilde{x}=h(x)$, and $\tilde{y}=h(y)$, where $h$ is given by Eq.~\eqref{eq:generalized_lamperti}.
    \State Generate $\overline{n}$ bridges obtained through Eq.~\eqref{eq:SDE_for_bridges2} with discretization unit 
    $\overline{\Delta t} = (t-s)/m$, $m\in\mathbb{N}^{\ge3}$, that connect the points ($s,\tilde{y}$) and ($t,\tilde{x}$). These discrete stochastic skeletons 
    $\{\overline{k}_i^{(j)}\}_{i=1,\dots,m}$, $j=1,\dots,\overline{n}$, 
    are used to evaluate Eq.~\eqref{eq:RN_discretization_F}. 
    \State Build the transformed drift function $\alpha_h$ through Eq.~\eqref{eq:transformed_drift_F}, which itself requires 
    the functions
    $h'$, $h''$, and $h^{(-1)}$.
    \State Using the bridges and  transformed drift from Steps 2 and 3, 
    evaluate the estimator
    $\tilde{L}^{H\mid K}_{\Delta t}(\{\overline{k}_i\}_{i=1,\dots,m})$
    of the likelihood ratio at 
    each one of the $\overline{n}$ generated bridges, according to Eqs.~\eqref{eq:RN_discretization_F} and~\eqref{eq:RN_discretization_F_V2}  with $\Delta t=\overline{\Delta t}$, yielding
    $\left\{\hat{L}^{{H\mid K}(j)}_{\overline{\Delta t}}\right\}_{j=1,\dots,\overline{n}}$.
    \State Estimate the propagator of the transformed process using Eq.~\eqref{eq:propagator_BCM_general}:
    \begin{equation*}
        \hat{\rho}^{H}_{t\mid s}(\tilde{x}|\tilde{y})
        =\frac{\rho^{K}_{s,t}(\tilde{x}|\tilde{y})}{\overline{n}}\sum_{j=1}^{\overline{n}}
        \hat{L}_{\overline{\Delta t}}^{H \mid K (j)}.
    \end{equation*}
    \State Obtain the BCMMCE of the propagator of the target process $X$
    by the change-of-variable formula Eq.~\eqref{eq:relation_between_propagators}: $$\hat{\rho}_{t|s;\overline{n},\overline{\Delta t}}^{X \; \text{BCM}} (x|y) 
    =\hat{\rho}^{H}_{t\mid s}(\tilde{x}|\tilde{y}) \,\mid h'(x) \mid.$$
   \end{algorithmic}
\end{algorithm}

Analogously to Algorithm~\ref{al:ACM}, when the target process $X$  
shares the same noise coefficient as the bridge process, that is, when  
$B(x)=\beta(x)$, Steps 1 and 6 become superfluous: in this 
case $h(x)=x$ yielding $H=X$, $\tilde{x}=x$, $\tilde{y}=y$, and 
$\hat{\rho}^{H}_{t\mid s}(\tilde{x}|\tilde{y})=\hat{\rho}^{X}_{t\mid s}({x}
|{y})$. 

Note that, both Algorithms \ref{al:ACM} and 
\ref{al:ACM_V2}
make use of the
the left-point discretization formulae of Sec.~\ref{sec:discretization_PI}, 
namely.
Eqs.~\eqref{eq:RN_discretization} 
and~\eqref{eq:RN_discretization_V2}. 
These two algorithms
can also be direclty adapted to the Riemannian trapezoidal
formulae Eqs.~\eqref{eq:RN_discretization_V3} 
and~\eqref{eq:RN_discretization_V4}.

\section{Parametric approximation to the best-in-class 
bridge measure}                                 \label{sec:Approximation}

Corollary~\ref{cor:optimallity} provides the optimal 
change-of-measure for the BCMMC 
approach: it is 
$\Pm_{X;s,y}^{t,x}$,
the bridge measure of the process $X$ whose
propagator $\rho^X_{t\mid s}(x|y)$ needs to be estimated.
Then, in Sec. \ref{sec:algorithms}, $X$ is transformed to $H$ through the
(generalized) Lamperti transform $h$.
In the practice, 
the bridge measure $\Pm_{H;s,y}^{t,x}$
is not accessible as its construction requires prior knowledge of the
unknown propagator itself.
We therefore pursue a more practical goal: finding the best approximation to the optimal bridge measure within a tractable, parametrized family of bridge measures. Concretely, recall from Sec.~\ref{sec:algorithm_general} that $K$ is 
the unconditional auxiliary process in Eq.~\eqref{eq:SDE_for_bridges}, which is then conditioned to generate the 
bridge process $\overline{K}$ in Eq.~\eqref{eq:SDE_for_bridges2}. 
We now embed $K$ into a larger class parametrized by the vector $\vec{\theta}$ belonging to some open set $\Theta$ of $\R^d$, $d \ge 1$.
It is precisely the drift function in Eq.~\eqref{eq:SDE_for_bridges} that is generalized and
the generalized sampling process $K_{\vec\theta}=K_{\vec\theta,[0,T]}=\{ K_{\vec\theta,t}
\}_{0\le t \le T}$
solves the SDE
\begin{equation}\label{eq:SDE_for_bridges_V2}
    dK_{\vec{\theta},t} = \alpha(K_{\vec{\theta},t};\vec{\theta}\,)dt+\beta(K_{\vec{\theta},t})\,dW_t,
    \quad 0 \le t \le T.
\end{equation}
The objective is to find the parameter $\vec\theta \in \Theta$ that brings 
$\Pm_{K_{\vec\theta};s,y}^{t,x}$
as close as possible to the optimal bridge measure
$\Pm_{H;s,y}^{t,x}$.
According to
Corollary~\ref{cor:optimallity} the unconstrained optimal bridge measure is obtained with original process, $K^*=H$, 
this forcing $L^{H\mid K^*}(K^*_{[s,t]})=1$, namely $l^{H\mid K^*}_{s,t} =0$, according to Eqs.~\eqref{eq_girs_l} and 
\eqref{eq_girs_F}. 
This tells that
$l^{H\mid K_{\vec{\theta}}}_{[s,t]}$ close to 0 indicates 
that the process $K_{\vec{\theta}}$ is a close approximation to
$K^*$. We thus propose the estimator of $\vec{\theta}$
given by
\begin{equation}                                \label{eq_lo_girs_F_optimal}
   \hat{\vec{\theta}} = \arg \min_{\vec{\theta} \in \Theta}  
   \mid l^{H\mid K_{\vec{\theta}}}_{s,t}\mid
\end{equation}
and we
define by $\hat{K} = K_{\hat{\vec\theta}}$ the resulting sampling process.
Clearly, if $K^*=H\in \{K_{\vec\theta}\}_{\vec{\theta}\in\Theta}$, then $\hat{K}=K^*$. Conversely, if $K^*\notin \{K_{\vec\theta}\}_{\vec{\theta}\in\Theta}$, then $\hat{K}$ 
is the best change-of-measure within the family $\{K_{\vec\theta}\}_{\vec{\theta}\in\Theta}$ 
in terms of variance reduction, as 
   $
   \text{var}\left(L^{H\mid K_\theta}\right)\ge\text{var}\left(L^{H\mid \hat{K}}\right)\ge \text{var}\left(L^{H\mid {K^*}}\right)
   = \text{var}\left(L^{H\mid H}\right)
   =0$, $\forall \theta \in \Theta$.
The minimization problem in Eq.~\eqref{eq_lo_girs_F_optimal} inherits the stochastic nature of the process $H$,
making $\vec{\theta}$ a random variable. 
A standard strategy that eases the problem is to tackle a deterministic version of the minimization, replacing 
$l^{H\mid K_{\vec{\theta}}}_{s,t}$ by its expectation in what is called the cross-entropy method; cf.~\cite{de2005tutorial}. This leads to
\begin{equation}                                \label{eq_lo_expected_girs_F_optimal}
   \vec{{\theta}}^\dag = \arg \min_{\vec{\theta} \in \Theta}   
   \E_{\Pm_{K_{\vec\theta};s,\tilde{y}}^{t,\tilde{x}}}
   \left[ \Big|
   l^{H\mid K_{\vec{\theta}}}_{s,t} \Big| \right],
\end{equation}
where $\tilde{y}= h(y)$ and $\tilde{x}=h(x)$, $h$ being the generalized
Lamperti transform in~\eqref{eq:generalized_lamperti}. We can re-express
$l^{H\mid K_{\vec{\theta}}}_{s,t}$
in terms of an integral of an It\^o integral as follows,
\begin{equation*} 
    l^{H\mid K_{\vec{\theta}}}_{s,t}
    =
    \int_s^t  \frac{\alpha_h\left(K_u\right)-\alpha\left(K_u;\vec{\theta}\right)}{\beta\left(K_u\right)}\,dW_u - \frac{1}{2} \int_s^t \frac{\left[\alpha_h\left(K_u\right)-\alpha\left(K_u;\vec{\theta}\right)\right]^2}{\beta^2\left(K_u\right)}\,du.
\end{equation*}
Thus, because It\^o's integrals have null expectation, we can simplify
Eq.~\eqref{eq_lo_expected_girs_F_optimal} to
\begin{equation}                                \label{eq_lo_expected_girs_F_optimal_V2}
   \vec{{\theta}}^\dag = \arg \min_{\vec{\theta} \in \Theta}  \E_{\Pm_{K_{\vec\theta};s,\tilde{y}}^{t,\tilde{x}}} \left[\int_s^t \frac{\left\{\alpha_h\left(K_u\right)-\alpha\left(K_u;\vec{\theta}\right)\right\}^2}{\beta^2\left(K_u\right)}\,du\right].
\end{equation}

Solving Eq.~\eqref{eq_lo_expected_girs_F_optimal_V2}  might still be cumbersome in the general case, hence we consider the special case of two-parameter families, i.e. $d=2$. 
Given that the integrand
in Eq.~\eqref{eq_lo_expected_girs_F_optimal_V2}
takes fixed values at the boundary points $s$ and $t$ under 
the bridge measure (of the expectation), we can easily approximate the 
integral  by the trapezoidal rule. Consequently, we simplify the 
optimization problem to the 
minimization, w.r.t. $\vec{\theta}$, of the
square bracket at $K_u = \tilde{y}$ plus
that bracket at $K_u = \tilde{x}$. We are lead to
solve w.r.t. $\vec{\theta}$ the the system of equations
\begin{equation}\label{eq:eq_theta}
\begin{cases}
    \alpha(\tilde{x};\vec{\theta})=\alpha_h(\tilde{x}),\\
    \alpha(\tilde{y};\vec{\theta})=\alpha_h(\tilde{y}),
\end{cases}
\end{equation}
which fully characterizes
$\vec{\theta}^\ddag$ when $d=2$ and $\tilde{x} \neq \tilde{y}$.
A key advantage of this approximation is that Eq.~\eqref{eq:eq_theta} admits an analytical solution in many useful cases. For example, when using OU bridges, i.e. Eq.~\eqref{eq:SDE_for_bridges} 
with $\alpha(x)=\overline{k}(\overline{\mu}-x)$ and $\beta(x)=\overline{D}$,
with $\overline{k},\overline{D}\in \mathbb{R}^{>0}$ and $\overline{\mu}\in\mathbb{R}$ being constants. Here $\vec{\theta} 
= (\overline{\mu}, \overline{k})$ and a 
solution to Eq.~\eqref{eq:eq_theta} reads
\begin{equation}    \label{eq:best_in_class_values_OU_bridges}
    \overline{k}^\ddag = \frac{\alpha_h(\tilde{x})-\alpha_h(\tilde{y})}{\tilde{x}-\tilde{y}},\quad \overline{\mu}^\ddag=\frac{\alpha_h(\tilde{y})}{\overline{k}^\ddag}+\tilde{y}.
\end{equation}
Since the evaluations of $\alpha_h$, $\tilde{x}$ and $\tilde{y}$ 
are already required in Algorithm~\ref{al:ACM_V2}, we obtain these two parameters of the bridge 
at essentially null computational cost.

\section{Numerical study}                                   \label{VII}

This section presents a numerical study of the accuracy of 
the BCMMC approach with three important target diffusions
and two different bridge processes.
The targets are the OU, the Cox-Ingersoll-Ross (CIR) and the 
stochastic logistic model (SL) processes; all of them defined in Table~\ref{tab:models}. Both the OU, whose formula 
is in Eq. \eqref{e33}, and CIR processes are paradigmatic models of extended use in finance, physics, ecology and have many other interdisciplinary applications, see e.g. 
\cite{Dornic2005,azaele_statistical_2016,George2023}. Their 
propagators are known and given in Appendix~\ref{AP_sec:propagators}. 
Note also that the OU process is a central 
stationary stochastic process and it
is the continuous time version of the autoregressive
time series of order one, known as AR(1).
Stationarity represents a fundamental property of both theoretical and applied relevance~\cite{gattobook22}. 
The SL model is widely used in ecology to describe logistic growth under environmental noise~\cite{Grilli2020}. Its propagator is 
not available in closed-form. This complicates both its theoretical analysis and the statistical inference~\cite{brigatti_growth-rate_2025,PRX_Aguilar_2026}. Thus, it provides non-tivial example for testing the BCMMC with relevant applications. 
\begin{table}[h]
    \centering
    \begin{tabular}{|c|c|c|c|c|}
    \hline
          \shortstack{Target diffusion\\$X$ in Eq.~\eqref{eq:general_SDE}}& \shortstack{Drift function\\$A(x)$} & \shortstack{Volatility function\\$B(x)$} & \shortstack{Domain of\\parameters} & 
          \shortstack{Propagator\\$\rho^{X}_{t\mid s}(x|y)$} in
          Eq. \eqref{eq:def_propagator}
         \\ \hline \hline
         
         Ornstein-Uhlenbeck (OU) & $k(\mu-x)$  & $D$ & $k,D\in \mathbb{R}^{>0}$, $\mu\in\mathbb{R}$ & Closed-form in Eq. \eqref{eq:OU_propagator} \\ \hline

         Cox-Ingersoll-Ross (CIR) & $k(\mu-x)$  & $D\sqrt{x}$ & $\mu,k,D\in \mathbb{R}^{>0}$ & Closed-form in Eq. \eqref{eq:propagator_CIR} \\
\hline
         Stochastic logistic (SL) & $k(\mu-x)x$  & $Dx$ & $\mu,k,D\in \mathbb{R}^{>0}$  & Closed-form not available\\ \hline
    \end{tabular}
    \caption{
    Description of the 3 target diffusions of
    the numerical examples}
    \label{tab:models}
\end{table}
Standard Wiener and OU bridges are considered in this section. The parameters
of the OU bridges are called as in Table \ref{tab:models} but with
the overbar: $\overline{k}$, $\overline{D}$ and $\overline{\mu}$.

In Sec. \ref{sec:OU_w_wiener_Bridges}, Wiener bridges 
are applied to all three OU, CIR and SL target processes.
In Sec. \ref{sec:varying_bridge}, OU bridges 
are used for OU and SL targets.

\subsection{Wiener bridges}\label{sec:OU_w_wiener_Bridges}

The BCMMC technique with Wiener bridges is studied in this section.
We follow the methodology for Wiener bridges of Sec. \ref{sec:algorithm}.

\subsubsection{Ornstein-Uhlenbeck target process}       \label{sec:OUtarget}

We first apply the BCMMC with the Wiener bridges of Eq.~\eqref{eq:W-bridge} to the OU process whose propagator is available:
see Table~\ref{tab:models} for the functions $A(x)$, $B(x)$ and for the
propagator.
 Since both the Wiener and the OU processes are additive (i.e. have constant volatility function), the Lamperti transform in Eq.~\eqref{eq:Lamperti_transformation} is merely
 a constant rescaling of the OU and the functions required for the new drift 
 function $A_h(x)$, in Eq. \eqref{eq:transformed_drift}
 read: $h(x)=x/D$, $h^{(-1)}(x)=Dx$, $h'(x)=1/D$ and $h''(x)=0$. 
We are interested in the relative errors  
    \begin{align}                                           \label{e44}
    \text{RE}=
    \frac{
    \rho^{X}_{t\mid s}(x|y)-\hat{\rho}_{t|s;\overline{n},\overline{\Delta t}}^{X, \; \text{BCM}}(x|y)}{
{\rho^{X}_{t|s}(x|y)}}.
\end{align}
Fig.~\ref{fig:OU_propagator} shows these relative errors,
as obtained from Algorithm~\ref{al:ACM}. The values of the 
parameters and of the bridges are: 
$\mu=D=k=y=t=1$, $s = 0$ and $x = 1.1$
There are in fact
$r = 100$ MC replications under the
same conditions: the central point of each bar plot in 
Fig.~\ref{fig:OU_propagator}, say $\overline{\text{RE}}$,
is the average of the $r$
relative errors and the extremities
of the bar plots provide the $95\%$-level asymptotic normal 
confidence intervals according to
$\overline{\text{RE}} \pm r^{-1/2} 1.96 S$,
$S$ being the standard deviation of the $r$ replications 
of $\overline{\text{RE}}$.
Both axes have 
logarithmic scale and we consider various numbers of
generated bridges $10^1 \le \overline{n} \le 10^5$ and
various values of the 
discretization unit
$10^{-3} \le \overline{\Delta t} \le 10^{-1}$. 
We see that the
largest relative errors are for $\bar{n} = 10^1$, 
remain essentially below $10\%$. This is remarkably small with the
given the minimal MC setup. As expected, the errors decrease as the number of bridges $\overline{n}$ increases and as 
the discretization unit $\overline{\Delta t}$ decreases. 
First, the relative errors tend to decrease proportionally 
to $\overline{n}^{-1/2}$, which is indicated by the dashed line in Fig.~\ref{fig:OU_propagator} and which reflects 
the statistical or sampling contribution to the relative error.
Then, the relative errors reach an asymptotic regime, 
which appears
determined by $\overline{\Delta t}$.
Achieving optimal accuracy requires balancing statistical or sampling 
error (controlled by $\overline{n}$) and discretization bias (governed by $\overline{\Delta t}$). 
For example, for the blue curve with $\overline{\Delta t} = 10^{-2}$ in Fig.~\ref{fig:OU_propagator}, increasing the number of realizations beyond $\overline{n} = 10^4$ does not yield a statistically significant improvement in precision, because the discretization-induced bias already constitutes the dominant contribution to the total relative error.
See also~\cite{aguilar_biased_2024} in this 
context.
\begin{figure}
    \centering
    \includegraphics[width=0.5\linewidth]{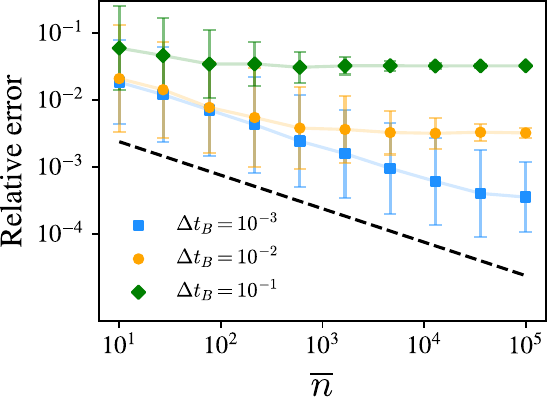}
    \caption{Relative error
    of BCMMCE for the OU target process 
    with the Wiener bridge obtained by
    Algorithm~\ref{al:ACM},
    with varying number of bridge samples ($\overline{n}$) and bridge
    discretization units ($\overline{\Delta t}$).
    Logarithmic coordinates on both axes.
    The dashed line at $\overline{n}^{-1/2}$ indicates the trend of the relative errors in absence of bias. The error bars indicate the $95\%$ confidence interval computed over $r=100$ replications. 
    Parameters values: $\mu=D=k=y=t=1$, $s = 0$ and $x = 1.1$.}
    \label{fig:OU_propagator}
\end{figure}

\subsubsection{Cox-Ingersoll-Ross target process}             \label{sCIR}

This section studies the BCMMCE of
the propagator of the CIR process. 
The functions $A(x)$ and $B(x)$ of Eq.~\eqref{eq:general_SDE}
are given in Table~\ref{tab:models}. The propagator is known and 
given in 
Eq. \eqref{eq:propagator_CIR} of
Appendix~\ref{AP_sec:propagators}.
It is a multiplicative process
in the sense that volatility function $B(x)$ is 
nonconstant. 
Now the Lamperti transform \eqref{eq:Lamperti_transformation} 
is essential for
the use of
Wiener bridges. 
The functions required by 
 Algorithm~\ref{al:ACM} read: $h(x)=2\sqrt{x}/D$, 
 $h^{(-1)}(x)=(Dx/2)^2$,
 $h'(x)=(D\sqrt{x})^{-1}$, $h''(x)=-(2 D x^{3/2})^{-1}$.
The values of the parameters and the bridges are: 
$t=k=1$, $D=0.5$, $\mu = y = 0.3$ and  $s=0$.
Fig.~\ref{fig:CIR} shows the perfect agreement between the 
BCMMCE with
Wiener bridge and the true propagator in Eq.~\eqref{eq:propagator_CIR}:
the density is shown in Fig.~\ref{fig:CIR}.a and the logarithm of
the density in Fig.~\ref{fig:CIR}.b.
The error bars show the $95\%$ confidence intervals computed over 
$r=100$ realizations.
The discretization unit is $\overline{\Delta t} = 0.01$.
Notably, using only $\overline{n}=10$ generated bridges yields highly accurate estimates of the propagator’s tail behavior.
Obtaining such precision with the CMCE in
\eqref{eq:MC_estimator_prop}
would necessarily require substantially more realizations, 
since we are accessing events with very low probability on the tails of the distribution. This last point will be illustrated numerically 
in Sec. \ref{sec:sl_w}, where the BCMMCE and the 
CMC estimator are compared.
\begin{figure}
    \centering
    \includegraphics[width=0.7\linewidth]{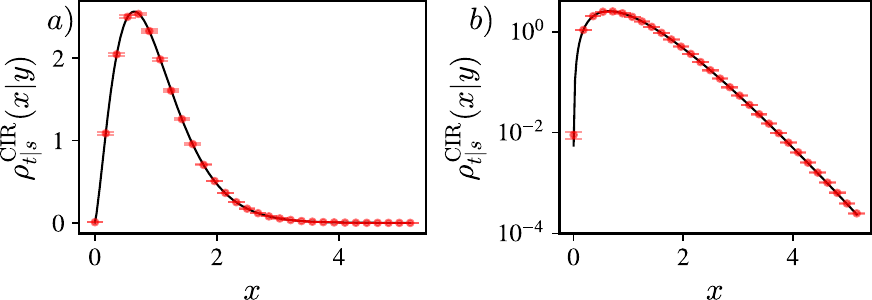}
    \caption{BCMMCE of CIR propagator with Wiener bridges (dots) and 
    exact formula (solid lines). Original scale in Fig. a and
    logarithmic $y$-axis is Fig. b. 
    The BCMMCE is computed with
    $\overline{n}=10$ bridges and discretization unit $\overline{\Delta t} = 0.01$.
    Error bars show $95\%$ confidence interval computed over $r=100$ realizations. Bridge and parameter values: $t=k=1$, $D=0.5$, $\mu = y = 0.3$ and  $s=0$.}
    \label{fig:CIR}
\end{figure}

\subsubsection{Stochastic logistic target process}              \label{sec:sl_w}

We now consider the SL target process
whose propagator is not available in closed-form expression. As before
we refer to
Table~\ref{tab:models} for the functions $A(x)$ and $B(x)$.
The functions required by Algorithm~\ref{al:ACM_V2} 
read in this case: $h(x)=\log(x)/D$, $h'(x)=(Dx)^{-1}$, $h''(x)=-(Dx^2)^{-1}$ and $h^{(-1)}(x)=e^{Dx}$.
Fig.~\ref{fig:SLM} shows the 
BCMMCE with Wiener bridges according to Algorithm~\ref{al:ACM_V2}
as well as the CMCE of Eq.
\eqref{eq:MC_estimator_prop}.
 We see a strong agreement between BCMMCE and CMCE. Notably, the 
 BCMMCE achieves the same level of precision around the mode using three orders of magnitude fewer sample paths than the CMCE. Moreover, it permits accurate estimation in regions of the state space that are 
 essentially
 inaccessible to crude MC, owing to the limited number of generations
 over such regions. 
 Indeed, with $n=10^4$ replications of the CMCE 
 cannot access the regions below $x=0.75$ and above $x=1.25$, while the 
 BCMMC method can access these two tails with only $\overline{n}=10$ realizations per end-point.
\begin{figure}
    \centering
\includegraphics[width=0.5\linewidth]{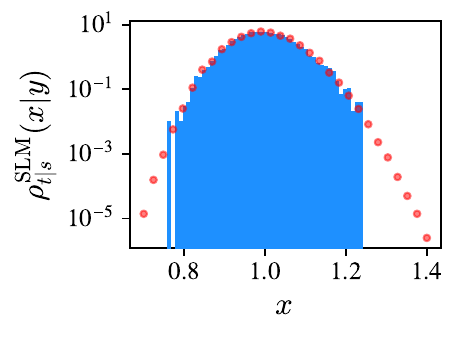}
    \caption{
    BCMMCE with $\overline{n}=10$ generations (dots) and CMCE 
    with $n=10^4$ generations (bars) of SL propagator. The two estimators  exhibit similar fluctuations over the central region but
    only
    the BCMMCE effectively accesses the tails, not the CMCE.
     The $95\%$ confidence intervals for BCMMCE computed with $r=100$ 
     realizations lie within the marker size. 
     Bridge and
     parameters values: $\mu=k=1$, $D=0.1$, $s=0$ and $t=1$.
}
    \label{fig:SLM}
\end{figure}

\subsection{Ornstein-Uhlenbeck bridges}\label{sec:varying_bridge}

We now illustrate the BCMMC approach with general bridges, according to Sec.~\ref{sec:algorithm_general}, focusing on OU bridges. 
Note that although the OU has constant volatility function $\beta(x)$,
just like the Wiener process, we cannot rescale the OU
and apply the methodology for Wiener bridges of Sec. \ref{sec:algorithm},
because it assumes null drift function $\alpha(x)$.
As mentioned, 
the method is particularly efficient 
when the bridge propagator is known. This allows for the exact
evaluation of
the propagator change-of-variable formula
Eq.~\eqref{eq:relation_between_propagators} and also 
for the exact generation of the bridges, according to~\cite{Majumdar2015,aguilar_sampling_2022}. 
Bridges of the OU process are obtained by conditioning the process $K$ in Eq.~\eqref{eq:SDE_for_bridges} with $\alpha(x)=\overline{k}(\overline{\mu}-x)$ and $\beta(x)=\overline{D}$, for given $\overline{k},\overline{D}\in \mathbb{R}^{>0}$ and $\overline{\mu}\in\mathbb{R}$. OU bridges can be generated exactly according to the method of~\cite{Majumdar2015,PRX_Aguilar_2026} and the propagator of $K$ is given in Eq.~\eqref{eq:OU_propagator}, upon 
replacing $k$, $\mu$ and $D$ by $\overline{k}$, $\overline{D}$ and $\overline{\mu}$, respectively. 

\subsubsection{Ornstein-Uhlenbeck target process}

Let us first start with the same problem of
Sec.~\ref{sec:OUtarget}, where the aim is to estimate
the OU propagator, but here with OU bridges. In particular, we focus on a simple, yet illustrative case in which the process $K$, 
used to generate bridges $K$, 
is almost identical to the target process $X$: 
we take $k=\overline{k}$, $D= \overline{D}$ but $\mu\neq \overline{\mu}$. 
Now Algorithm~\ref{al:ACM_V2} simplifies substantially: 
as explained immediately after Algorithm~\ref{al:ACM_V2},
the generalized Lamperti transform is the identity when 
$B(x)=\beta(x)$.

Fig.~\ref{fig:OU_OU_bridges} shows the standard deviation of the 
BCMMCE with OU target and OU bridges, for
varying values of the ratio of bridge over target
means, $\overline{\mu}/\mu$. 
The value of the bridges and of the model parameters are: 
$y = \mu=k=D=1$, $x = 0.1$, 
$s=0$ and $t=1$.
The values of the simulation are:
$\overline{\Delta t} = 0.01$ and $\overline{n}=10^3$.
As claimed by Corollary~\ref{cor:optimallity}, the standard deviation 
tends to zero as the process $K$ approaches 
the target process $X$ (i.e. as $\overline{\mu}/\mu \to 1$).
\begin{figure}
    \centering
\includegraphics[width=0.5\linewidth]{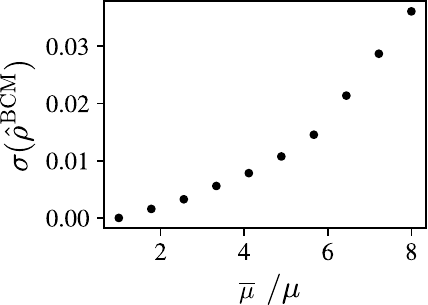}
    \caption{Standard deviation of the BCMMCE of OU propagator with
    OU bridges 
    with $\overline{n}=10^3$ generations and
    discretization unit
    $\overline{\Delta t} = 0.01$
    for varying ratio of target over bridge means
    $\overline{\mu}/\mu$. 
    Target and bridge OU processes differ only through
    their means, $\mu$ and $\overline{\mu}$ respectively. 
    Values of bridges and parameters: 
    $y = \mu=k =\overline{k}  =D =\overline{D} =1$, $x = 0.1$,  
    $s=0$ and $t=1$.
}
    \label{fig:OU_OU_bridges}
\end{figure}

\subsubsection{Stochastic logistic target process}

We now turn to the SL target process
already considered with Wiener bridges 
in Sec.~\ref{sec:sl_w}
but here with OU instead of Wiener bridges.
We remind that the SL propagator is not available.
We consider various OU bridges differing by the
 values of the parameter $\overline{k}$. 
 Algorithm~\ref{al:ACM_V2} requires
 the generalized Lamperti transform that renders 
 the SL process additive: it is already given in 
 Sec.~\ref{sec:sl_w}.
Fig.~\ref{fig:SLM_OU_bridges} shows the mean square error (MSE) between the BCMMCE and a surrogate to the exact, which is the CMCE with 
$n=10^5$ replications, which is  computed across 
parameter values of the bridges $0 \le \overline{k}\le 15$
in Fig.~\ref{fig:SLM_OU_bridges}.a
and across
the domain of propagator arguments $0.7 \le x
\le 1.4$ in Fig.~\ref{fig:SLM_OU_bridges}.b.
The $y$-axis of both graphs are in logarithmic scale.
The MSE is obtained by 
replicating $r=30$ times the same simulation scheme,
with the formula
${\rm MSE} = r^{-1} \sum_{j=1}^{r} \left\{ \hat{\rho}_{t|s;n}^{X \; {\rm BCMMC} \; (j)} (x|y) - \hat{\rho}_{t|s;n}^{X \; {\rm CMC}} (x|y) 
\right\}^2$.
Since the auxiliary process $K$ cannot coincide with the target process $X$
and therefore the BCMMCE with minimal variance cannot be obtained, 
according to Corollary~\ref{cor:optimallity},
$\overline{k}$ can be taken tuning parameter for the minimization of the MSE.
Fig.~\ref{fig:SLM_OU_bridges}.a shows that
 the MSE is minimized for $\overline{k}$ close to $k=3$, which is marked by a 
 dotted line.
 Notably, a poorly chosen value of $\overline{k}$ can perform substantially worse than the driftless Wiener bridge, obtained with 
 $\overline{k}=0$, this by up to three orders of magnitude in this example. 
The values for the bridges
and of the model
parameters are: $y=x_0=\mu=1$, $k=3$, $D=0.3$, $s=0$, $t=0.75$, $\overline{\mu}=-0.05$.
The values of the MC study are
$\overline{\Delta t}=0.01$ and $\overline{n}=1$. There are 
$r=30$ same replications per value of $\overline{k}$ and
the CMC reference is based on $n=10^5$ realizations.
 Matching qualitative features of the target $X$ (e.g. mean-reversion) may 
 not lead to an efficient auxiliary process $K$. What 
 appears more important is the similarity between the
 auxiliary and target drift functions, $\alpha(x)$ and $A(x)$
 respectively, as suggested by Corollary~\ref{cor:optimallity}.
 Fig.~\ref{fig:SLM_OU_bridges}.b compares the BCMMCE obtained 
 with the best- and the worst-performing values of $\overline{k}$, 
 against the reference curve of CMCE, for $0.7 \le x \le 1.4$.
 The best and worst $\overline{k}$ are identified in
 Fig.~\ref{fig:SLM_OU_bridges}.a.
The $y$-axes of the two graphs are in logarithmic scale.
While the BCMMCE with best value of $\overline{k}$ is shown 
by green dots. They reproduce the target propagator accurately across the whole sampled range. The
BCMMCE with worst $\overline{k}$ and shown by red squares. They 
display erratic fluctuations, especially away from the mode of the distribution.
\begin{figure}
    \centering
\includegraphics[width=0.9\linewidth]{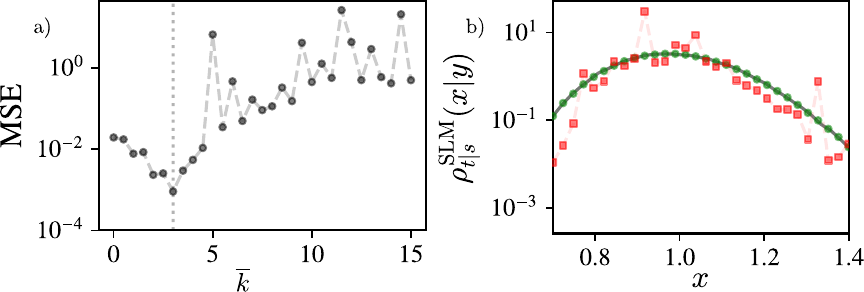}
    \caption{BCMMCE of SL propagator with OU bridges.
    for varying values of the bridge parameter $\overline{k}$
    and propagator argument $x$.
    Fig. a shows the MSE between BCMMCE and CMCE taken as reference, 
    for various values of $\overline{k}$ and
    averaged over $0.9 \le x \le 1.4$. The dotted line at 
    $\overline{k}=k=3$ corresponds to the minimal MSE.
    Fig. b 
    shows the BCMMCE
    using the best-performing (green circles) and worst-performing (red squares) values of $\overline{k}$, taken from from Fig. a, 
    compared with the CMCE reference (black curve). Values for 
    bridges and parameters: $y=x_0=\mu=1$, $k=3$, $D=0.3$, $s=0$, $t=0.75$, $\overline{\mu}=-0.05$.
    Values of the MC scheme:
    $\overline{\Delta t}=0.01$, $\overline{n}=1$, averaged over $r=30$ 
    same replications per $\overline{k}$ and $n=10^5$ CMC replications.}
    \label{fig:SLM_OU_bridges}
\end{figure}

\subsection{Approximation to best-in-class Ornstein-Uhlenbeck bridge} \label{sec:best_in_class_approx}

Sec.~\ref{sec:varying_bridge} illustrates the ability of
parametrized bridge measures to compute propagators and the existence of a
best-in-class parameter choice, which
minimizes variance and thus statistical error of the BCMMCE,
within the bridge family.
We now test the 
analytical approximation presented in Sec.~\ref{sec:Approximation}
for finding such best-in-class bridge.
For this purpose, we consider the CIR and SL propagators analyzed in previous sections and we use OU bridges.

We first compute the BCMMCE for the CIR process using OU bridges, evaluating relative errors over a 
two-dimensional grid of parameters $(\overline{\mu}, \overline{k})$,
which is
centered about the best-in-class approximation
parameters
$( \overline{\mu}^\ddag,  \overline{k}^\ddag )$, given in 
Eq.~\eqref{eq:best_in_class_values_OU_bridges}. 
Because the exact propagator is known, relative errors can be computed.
At each pair $(\overline{\mu}, \overline{k})$
 on the grid, the BCMMCE 
 is evaluated using only $\bar{n}=1$ replication
 and 
 $10^3$ such BCMMCEs are generated.
 Their average allow us to obtain
 the relative error in Eq. \eqref{e44}.
This is done for all pairs
 $(\overline{\mu}, \overline{k})$
 and these pairs are then ordered according to their values of the 
 relative error: 
 rank 0 corresponding to the pair with smallest relative error
 and rank $10^3-1$ to the pair with largest one.
 The grid includes 
 both the parameters of our
 best-in-class approximation
$( \overline{\mu}^\ddag,  \overline{k}^\ddag )$
 and the parameters of the Wiener bridges 
 $(\overline{\mu},0)$. 
 Fig.~\ref{fig:SLM_optimal_OU_bridges}.a shows the relative errors
 w.r.t. to the ranks.
The dashed gray
lines mark the value obtained using Wiener bridges ($\overline{k}=0$), the
dotted red lines mark the value obtained with the best-in-class approximation
in Eq.~\eqref{eq:best_in_class_values_OU_bridges} and the solid blue lines
mark the smallest value found in the grid, namely the one
at rank 0.
The values of the parameters are the following: 
$\mu=1$, $k=1$, $D=0.1$, $X_0=1$, $s=0$,
$t=1/k=1$, $y=1$, $x=\mu+10^{-4}$, 
$\overline{n}=1$, $\overline{\Delta t}=0.01$. The
grid has $31\times31$ values of $(\overline{\mu}, \overline{k})$ 
within
$
[\overline{\mu}^\ddag-20,\,
\overline{\mu}^\ddag+20] \times [\overline{k}^\ddag-3,\,
\overline{k}^\ddag+3]$.
Fig.~\ref{fig:SLM_optimal_OU_bridges}.a shows that the best-in-class approximation substantially improves on the Wiener bridge estimator, which 
 already demonstrates 
 good accuracy in Sec.~\ref{sCIR}. However, it does not attain the smallest error found within the grid. This is consistent with the fact that our 
 solution in Eq. \eqref{eq:best_in_class_values_OU_bridges}
 represents only an easy-to-evaluate approximation to the true 
 best-in-class parameters.
Fig.~\ref{fig:SLM_optimal_OU_bridges}.b shows the results
of the analogue analysis for the
SL propagator.
The values of the parameters are as follows:
$\mu=1$, $k=3$, $D=0.3$, $X_0=1$, $s=0$,
$t=0.75$, $y=1$, 
$x=0.95$, $\overline{n}=1$, $\overline{\Delta t}=0.01$. The grid has
$31\times32$ values of $(\overline{\mu}, \overline{k})$ spanning
$[-0.3, 0.3] \times [0.5, 8.0]$, which is augmented with
$\overline{k}=0$ for the Wiener bridge and with the analytical
approximation $(\overline{\mu}^\ddag, \overline{k}^\ddag)$ itself.
At each pair $(\overline{\mu}, \overline{k})$
 on the grid
 and at the points of augmentation, the BCMMCE 
 is evaluated using only $\bar{n}=1$ replication
 and 
 $10^3$ such BCMMCEs are generated. Then the sample variance these
 values is computed.
The ranking of the
parameter pairs 
$(\overline{\mu}, \overline{k})$ 
is now constructed for the variance instead of
relative error. This is because
the closed-form expression for the propagator,
used to compute the relative error, is not available here, for the SL model.
Instead of the relative error, the variance serves as a proxy
for the statistical error of the BCMMCE. 
The various lines have the same identifications as in
in Fig.~\ref{fig:SLM_optimal_OU_bridges}.a.
As before, 
the best-in-class approximation
in Eq.~\eqref{eq:best_in_class_values_OU_bridges}
clearly outperforms the Wiener bridge BCMMCE, which is itself already
a strong baseline. In this case, however, 
the variance of our
approximation with
$(\overline{\mu}^\ddag, \overline{k}^\ddag)$
coincides with the smallest variance found over the full grid, indicating
that the approximation in
Eq.~\eqref{eq:best_in_class_values_OU_bridges}
is very close to the true best-in-class parameters
for this transition.
In  Fig.~\ref{fig:SLM_optimal_OU_bridges}.b,
the analytical approximation 
$( \overline{\mu}^\ddag,  \overline{k}^\ddag )$
coincides with the 0-rank value, while in
 Fig.~\ref{fig:SLM_optimal_OU_bridges}.a
 it remains above it, though still several orders of magnitude
below the Wiener bridge baseline.

While the shapes of the lines of
Figs.~\ref{fig:SLM_optimal_OU_bridges}.a and ~\ref{fig:SLM_optimal_OU_bridges}.b
do depend on the specific
parameters of the target process, we
consistently observed across different tests that the best-in-class
approximation in Eq.~\eqref{eq:best_in_class_values_OU_bridges}
improves the efficiency of BCMMCE by several orders of
magnitude relative to using Wiener bridges.
\begin{figure}
    \centering
\includegraphics[width=0.9\linewidth]{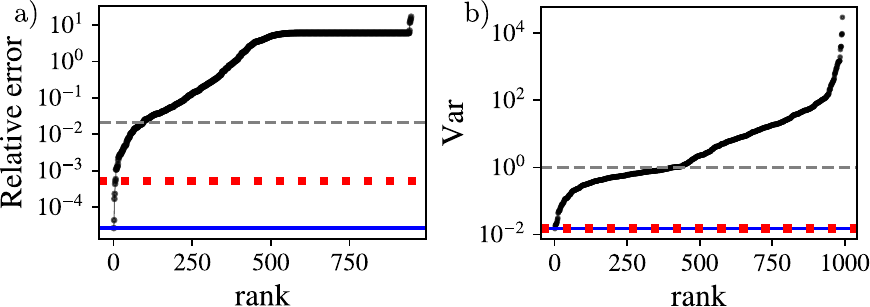}
   \caption{
   Approximation of the best-in-class parameters for BCMMCE of
   CIR and SL processes
   with OU bridges. 
   Fig. a shows the relative error of the BCMMCE with the CIR model
over a grid of bridge parameters
$(\overline{\mu}, \overline{k})$ centered at
$( \overline{\mu}^\ddag,  \overline{k}^\ddag )$, ranked in increasing order
of relative error of BCMMCE: 
rank 0 for smallest relative error and rank $10^3-1$ for largest.
Fig. b displays the variance of the BCMMCE with the SL model over a grid of bridge parameters
$(\overline{\mu}, \overline{k})$,
ranked increasing order of variance of BCMMCE: rank 0 for smallest
variance and rank above $10^3$ for largest.
In Fig. a and b
the dashed gray
lines mark the value obtained using Wiener bridges ($\overline{k}=0$), 
the
dotted red lines mark the value at
$( \overline{\mu}^\ddag,  \overline{k}^\ddag )$ 
and the solid blue lines
mark the smallest (or 0-rank) value found in the sampled grid. 
Parameter values for Fig. a (CIR): $\mu=1$, $k=1$, $D=0.1$, $X_0=1$, $s=0$,
$t=1/k=1.0$, $y=1$, $x=\mu+10^{-4}$, $\overline{n}=1$, $\overline{\Delta t}=0.01$,
grid of $31\times31$ values of $(\overline{\mu}, \overline{k})$ spanning 
$[\overline{\mu}^\ddag-20,\,
\overline{\mu}^\ddag+20] \times [\overline{k}^\ddag-3,\,
\overline{k}^\ddag+3]$, each point with $10^3$ 
independent realizations of the BCMMCE.
Parameters values for Fig. b (SL): $\mu=1$, $k=3$, $D=0.3$, 
$X_0=\mu=1$, $s=0$,
$t=0.75$, $y=1$, $x=0.95$, $\overline{n}=1$, 
$\overline{\Delta t}=0.01$, grid of
$31\times32$ values of $(\overline{\mu}, \overline{k})$ spanning
$[-0.3, 0.3] \times [0.5, 8.0]$, augmented with
$\overline{k}=0$ for the Wiener bridge case and with
$( \overline{\mu}^\ddag,  \overline{k}^\ddag )$, each 
point
with $10^3$ 
independent realizations of the BCMMCE
}
    \label{fig:SLM_optimal_OU_bridges}
\end{figure}

\subsection{Comparing path integral discretizations} \label{sec:compare_PI_discretizations}

Two discretization schemes for the evaluation of a path 
integral appearing in the BCMMCE are
presented in Sec.~\ref{sec:discretization_PI}. This section
compares these two schemes.
It provides the relative errors of the 
BCMMCE under the CIR model and OU bridges,
using the left-point approximation of the Radon-Nikodym derivative of Eqs.~\eqref{eq:RN_discretization} and \eqref{eq:RN_discretization_V2}
(which was used in \cite{aguilar_biased_2024}) and the Riemannian
trapezoidal discretization of Eqs. \eqref{eq:RN_discretization_V3} and
\eqref{eq:RN_discretization_V4}

The values of the parameters 
are the following: $\mu = k = 1$, $D = 0.1$, $X_0=1$, $s = 0$, $y = \mu = 1$, 
$x = \mu + 0.1 = 1.1$ and $\overline{n} = 10^4$ bridges per point.
The parameters of the OU bridge
are equal to the best-in-class approximation 
$( \overline{\mu}^\ddag,  \overline{k}^\ddag )$, as given in 
Eq.~\eqref{eq:best_in_class_values_OU_bridges}. 
we refer to Table~\ref{tab:models} for 
the functions $A(x)$, $B(x)$ and for the
exact propagator, which 
allows for the computation of the
relative errors.
Fig.~\ref{fig:discretization}
shows the relative error of the two BCMMCE of the CIR propagator as 
function of the time window $t$ (given $s=0$).
The black circles mark the left-point discretization of the Radon-Nikodym derivative, as given in 
Eqs.~\eqref{eq:RN_discretization} 
and \eqref{eq:RN_discretization_V2}
(and used in \cite{aguilar_biased_2024}). 
The red squares mark the new Riemannian trapezoidal discretization, as given in Eqs.~\eqref{eq:RN_discretization_V3} and.~\eqref{eq:RN_discretization_V4}.
The bridge discretization step is kept at the fixed small fraction of the time window $\overline{\Delta t} = t/100$, so 
the number of bridge points $\bar{n}$ 
is approximately the same across all sampled values within the window of 
size $t$.

As expected, the Riemannian trapezoidal approximation outperforms the left-point scheme at all windows times $t$.
Fig.~\ref{fig:discretization} shows also 
that the relative errors of both estimators grow steadily w.r.t. the time window $t$. 
This is a direct consequence of discretization bias: since the bridge step is kept at a fixed fraction of the window ($\overline{\Delta t} = t/100$), 
it grows in absolute terms w.r.t. $t$, so each bridge point carries a coarser local approximation and the resulting bias accumulates over the whole path. This growth is, however, of limited practical relevance: as $t$ increases, the propagator $\rho^{\mathrm{CIR}}_{t|s}(x|y)$ converges to the stationary distribution of the process (a Gamma distribution for the CIR model), 
which no longer depends on the conditioning point $y$ and which 
is available in closed-form for all diffusions solving the SDE in  Eq.~\eqref{eq:general_SDE}. In this regime the target quantity is therefore trivial to compute directly, without any need of the BCMMCE.

\begin{figure}
    \centering
    \includegraphics[width=0.5\linewidth]{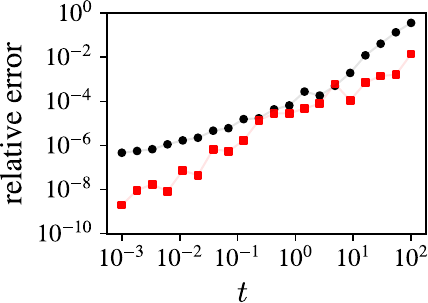}
    \caption{Relative error of the BCMMC estimator of the CIR propagator with
    OU bridges w.r.t. time window $t$. 
     Black circles: left-point discretization of the integral in
     likelihood ratio. Red squares: Riemannian trapezoidal discretization
     of the integral.
  Parameters fixed through the best-in-class approximation
    The bridge discretization step is kept at a fixed fraction of the time window, $\overline{\Delta t} = w/100$, so the number of bridge points is approximately the same across all sampled values of $dt$. Parameters used: $\mu = k = 1$, $D = 0.1$, $s = 0$, $y = x_0 = \mu$, $x = \mu + 0.1$, 
    $\overline{\Delta t} = t/100$,
    $\overline{n} = 10^4$ bridges per point, 
$( \overline{\mu},  \overline{k} )=
( \overline{\mu}^\ddag,  \overline{k}^\ddag )$ given by Eq.~\eqref{eq:best_in_class_values_OU_bridges}.   }
    \label{fig:discretization}
\end{figure}

\section{Discussion}                            \label{sVIII}

The first objective of this work was to place the BCMMC 
method for propagator estimation within a rigorous mathematical formalism.
This objective was achieved through two theorems.
First, Theorem~\ref{th:RN_Q_Q^B_rand} establishes the absolute continuity of 
the measure induced by $Y$ with a fixed point
($\mathbb{P}_{Y;s,a}$)
w.r.t.
the bridge process with a ramdom arrival point 
(${\tilde{\mathbb{P}}}_{Y;s,a}^{t,B}$)
and derives 
the corresponding Radon-Nikodym derivative
($d\mathbb{P}_{Y;s,a}/d{\tilde{\mathbb{P}}}_{Y;s,a}^{t,B} 
\mid_{[s,t]}$).
Second, Theorem \ref{tpropag} provides the bridge change-of-variable
formula linking
the two propagators
 $\rho^{X}_{t|s} \left(b \mid y \right)$  
 and
 $\rho^{Y}_{t\mid s}(b|y)$. With
Theorems~\ref{th:RN_Q_Q^B_rand} and \ref{tpropag} we could
build the mathematical foundation 
Algorithms \ref{al:ACM} and \ref{al:ACM_V2} of BCMMC, that are
the practical contribution of the article.
These algorithms rely also on a third essential mathematical result: the (generalized) Lamperti transform.
The
BCMMC can be applied to an arbitrary one-dimensional diffusion 
$X$ by means of the (generalized) Lamperti $H_t=h(X_t)$. 
The transformed process 
shares the same volatility function of 
the sampling process $K$, whose propagator is known.
This implies that the required absolute
continuity of the measures induced by $H$ 
and $K$, thus providing the likelihood ratio
of the respective bridges 
$d \Pm_{H;s,y}  / d \Pm_{K;s,y} \mid_{[s,t]}$.
To broaden the applicability of the BCMMC technique, we also derived an approximate optimal choice of bridge parameters within a chosen 
parametric
class, optimality being in the sense of
variance reduction. We developed a new specific discretization scheme for the path integrals appearing in the BCMMCE. These improvements reduce both the statistical and bias errors, thereby increasing the computational efficiency of the method. Finally, we presented several numerical illustrations that demonstrate the usefulness of the approach.  A Python implementation of the method, powered with \emph{NJIT} for fast computations is shared at Ref.~\cite{Aguilar2026BCMcode}. Such code is easy to adapt to solve arbitrary one-dimension diffusions.

Further contributions are the following. First,
we identified two key advantages of the BCMMC over the CMC for propagator estimation. First, BCMMC removes the Dirac delta from the expectation defining the propagator, upon which CMC is based, thereby avoiding the need for the state space discretization of CMC.
Second, Corollary~\ref{cor:optimallity} shows that a measure that minimizes 
the variance of the importance sampling propagator estimator is a bridge measure, suggesting that bridge processes are particularly efficient for propagator estimation. Intuitively, BCMMC can achieve greater precision than CMC when estimating probability densities of specific transitions 
(from one specific point 
$(s,y)$ to some other one $(t,x)$, in $[0,T]\times \R$)
because it concentrates the sampling effort on the target transition, whereas CMC distributes samples across the entire time-state space. Consequently, BCMMC and CMC may be viewed as complementary approaches suited to different tasks: while BCMMC is particularly appropriate for estimating propagators associated with specific transitions, typically being rare events (see \cite{aguilar_sampling_2022,Aguilar_Gatto_2024}), CMC is preferable when one seeks to evaluate the propagator over broader regions of the time-state space, especially those surrounding the most likely states.
This approach is also readily applicable to processes in which the functions $A$ and $B$ depend explicitly on time, since Lamperti transforms can be further generalized to this setting (cf. e.g.~\cite{Di_Bello_2025}). However, extending this approach to multidimensional processes is more restrictive. When the volatility is matrix-valued, there is no general analogue of the one-dimensional Lamperti transform that can map an arbitrary diffusion matrix to a prescribed constant form. Consequently, the applicability of BCMMC to multidimensional processes is limited to specific classes of processes for which a suitable auxiliary bridge process can be efficiently sampled and whose measure is absolutely continuous with respect to that of the multidimensional target process. A detailed assessment of the limitations and possible extensions of the method in the multidimensional setting is an important direction for future research.

We developed an improvement on both optimality (that reduces statistical errors of the method) and discretization schemes (that reduces bias). Nevertheless, an exhaustive analysis of these two refinements are outside the scope of this work and provide another important direction for future works. For example, a simple enhancement of the discretization can be achieved by using the general trapezoidal rule constructed on a non-uniform integration lattice, whose density is adapted to the local smoothness properties of the function $q$ to be integrated.

\section{Code availability}

Codes used to generate the figures of this manuscript are available at Ref.~\cite{Aguilar2026BCMcode}.

\section{Acknowledgements}
 J.A. acknowledges financial support from Project UNCOMMON  ``Unveiling the ecological dynamics of taxonomic and functional composition of microbial communities'', funded by the Italian Ministry of  Research (MUR), Future in Science - FIS Programme (CUP J53C23002290001).

\section{Appendix}
\subsection{Two propagators}\label{AP_sec:propagators}

For the sake of self-containment, we report here the analytical 
formulae of the OU and CIR propagators used throughout this work, which serve either to compute relative errors or to be employed in the BCMMC formulae. 
The OU and CIR processes are defined through the drift and volatility functions of Table~\ref{tab:models}. 
The propagator of the OU process reads
\begin{equation}\label{eq:OU_propagator}
    \rho^{X}_{t|0}(x|y) = \phi\left(x;\mu+(y-\mu)e^{-kt},\frac{D}{\sqrt{k}}\sqrt{1-e^{-2kt}}\right),
\end{equation}
with
$\phi(\cdot;\cdot,\cdot)$ the
Gaussian density of Eq. ~\eqref{eq:gauss}.

The propagator for the CIR model is, 
\begin{equation}\label{eq:propagator_CIR}
    \rho^{X}_{t|0}(x|y) = c\,e^{-u-v}\left(\frac{u}{v}\right)^{q/2}I_q\left(2\sqrt{uv}\right),
\end{equation}
where $q = 2\mu k/D^2-1$, $u = cye^{-kt}$, $v = c x$, $c = 2k/[D^2(1-\exp(-kt))]$ and $I_q(\cdot)$ 
modified Bessel
function $I$ of order $q\in \{0,1,\ldots\}$;
cf. 9.6.19 at p. 376 of \cite{abramowitz1948handbook}.

We note that both the CIR and OU processes are stationary is the sense that $\rho^X_{t\mid s}(x\mid y) = \rho^X_{t-s\mid 0}(x\mid y)$.

%
\bibliography{refs}

\end{document}